\documentclass[11pt]{article}
\usepackage{geometry}                % See geometry.pdf to learn the layout options. There are lots.
\usepackage{graphicx}
\usepackage{amssymb}
\usepackage{amsmath}
\usepackage{epstopdf}
\usepackage[title]{appendix}
\usepackage{amsthm}

\newtheorem{proposition}{Proposition}

\newtheorem{remark}{A/I Remark}[]

\usepackage{multirow}
\usepackage{setspace}
\usepackage{authblk}
\newcommand{\ind}{\rotatebox[origin=c]{90}{$\models$}}

\usepackage[plain]{algorithm2e}
\usepackage{tikz}

\usepackage[all,cmtip]{xy}
\usepackage{lineno}
\title{Inference on the structure of gene regulatory networks}
\author[1,2,*]{Yue Wang}
\author[3]{Zikun Wang}%\thanks{E-mail address: wangzikun329@gmail.com (Z. Wang)}}
\affil[1]{Department of Computational Medicine, University of California Los Angeles, Los Angeles, California, United States of America}
\affil[2]{Institut des Hautes \'Etudes Scientifiques (IH\'ES), Bures-sur-Yvette, Essonne, France}
\affil[3]{Laboratory of Genetics, The Rockefeller University, New York, New York, United States of America}
\affil[*]{Corresponding author. E-mail address: yuew@g.ucla.edu (Y. W.). ORCID: 0000-0001-5918-7525}
\date{}                                           % Activate to display a given date or no date
\begin{document}
\maketitle

\begin{abstract}
In this paper, we conduct theoretical analyses on inferring the structure of gene regulatory networks. Depending on the experimental method and data type, the inference problem is classified into 20 different scenarios. For each scenario, we discuss the problem that with enough data, under what assumptions, what can be inferred about the structure. For scenarios that have been covered in the literature, we provide a brief review. For scenarios that have not been covered in literature, if the structure can be inferred, we propose new mathematical inference methods and evaluate them on simulated data. Otherwise, we prove that the structure cannot be inferred. 
\end{abstract}

\smallskip
\noindent \textbf{Keywords.} 

\noindent Inference; Gene regulatory network; Independence; Differential equation.

\

\noindent \textbf{Frequently used abbreviations:} 

\noindent GRN: gene regulatory network.

\noindent DAG: directed acyclic graph.

\noindent ODE: ordinary differential equation

\section{Introduction}
\label{intro}
In living cells, most genes are transcribed to mRNAs and then translated to proteins. Some proteins serve as transcription factors that regulate the transcription of other gene(s). For example, in fruit flies, CLOCK (CLK) and CYCLE (CYC) are transcription factors that activate the transcription of \textit{period} (\textit{per}) and \textit{timeless} (\textit{tim}), the protein products of which in return inhibit CLK/CYC. Such transcription-translation negative feedback loop plays a central role in circadian rhythm \cite{patke2020molecular}. Small molecules also participate in gene regulation. Retinoic acid inhibits Fgf8 in mice \cite{wang2020biological}. For \textit{E. coli}, LacI protein inhibits the expression of \emph{lac} operon, unless lactose level is high and glucose level is low \cite{muller2013lac}. The regulation of gene expression is a central topic in biology.

Broadly speaking, a gene regulatory network (GRN) consists of various molecular regulators that govern gene expression, namely, levels of mRNAs and corresponding proteins \cite{cunningham2015mechanisms}. Such regulators interact with each other and form a complicated network. See Fig. \ref{grnex} for an example of GRN in \textit{E. coli}. Here the expression of one gene can activate or inhibit the expression of another gene. In this paper, the major focus is to determine the regulatory relationship between genes. Regulations involving non-genetic factors, or even other regulation networks (not necessarily biological), can be determined with similar methods.

\begin{figure}
	\center
	$\xymatrix{
		rpoD\ar[rr]\ar@/^/[rrdd]\ar@/^/[dd]&&rpoS\ar[dd]\\
		&&\\
		umuDC\ar@/^/@{-->}[uu]\ar@/^/@{-->}[rr]&&recF\ar@/^/[ll]\ar@/^/[lluu]
	}$\\
	\caption{An example of GRN in \textit{E. coli} \cite{bansal2006inference}. Each vertex is a gene (some genes are omitted). A solid arrow means activation, and a dashed arrow means inhibition.}
	\label{grnex}
\end{figure}
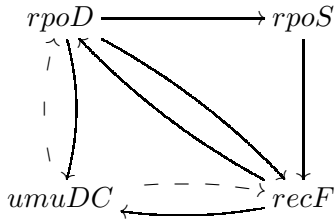

Most subjects of a GRN are large molecules confined within living cells. Therefore, it is difficult or even impossible to directly verify a regulatory network among a group of genes with biochemical methods. On the other hand, a large amount of data that are related to GRNs have been collected and documented, such as single-cell observational gene expression data \cite{cao2020human} and bulk level interventional phenotype data \cite{galluzzi2012prognostic}. The central question of this paper is to infer the GRN structure using experimental data (``reverse-engineering''). We aim at building a complete framework and fill in the blanks, not reviewing all important existing methods. Besides gene regulation, the problem of inferring network structure with experimental data is extensively studied in other fields \cite{brugere2018network,breza2020using,mateos2019connecting}. We consider two major situations:

(1) Consider genes $V_1,\ldots,V_n$. We can measure the expression levels of these genes (mRNA or protein count). The goal is to reconstruct the GRN that contains a specific gene, such as $V_1$. This means that we need to exclude irrelevant genes that do not belong to this GRN. Besides, for any two genes $V_i,V_j$ that belong to this GRN, we need to study whether the expression of $V_i$ directly regulates the expression of $V_j$. If yes, we determine whether the regulation is activation or inhibition.

(2) Consider genes $V_1,\ldots,V_n$ and a phenotype $V_0$ (growth rate etc.). Some genes can affect the level of this phenotype, and we assume this phenotype does not affect the expression of genes. We can measure the level of this phenotype, but not the expression levels of genes. Nevertheless, we can intervene with any genes, so as to change (decrease in general) their expression levels. By adding intervention on each gene $V_i$ and observe the phenotype $V_0$, we can exclude irrelevant genes. Thus we can stipulate that each $V_i$ can directly or indirectly affect $V_0$. Similarly, the goal is to reconstruct the GRN that contains $V_0$. We study whether the expression of $V_i$ activates or inhibits the expression of $V_j$.

Various data types can be used to infer GRN structures. Different data types require different mathematical inference tools. For certain data types, there have been numerous studies in determining GRN structures \cite{bansal2006inference,penfold2011infer,emmert2012statistical,wang2019comparative,hurley2012gene,xiong2004identification,yu2004advances,shmulevich2002gene}. We propose a unified framework to treat the GRN inference problem. Depending on possible data types that can be used to infer the GRN structure, we classify the GRN inference problem into 20 scenarios. For each scenario, we discuss what can be inferred about the GRN structure, and what assumptions are required. For scenarios that have been extensively studied, we only introduce a few representatives to show what we can infer in these scenarios. Readers may refer to some review papers for more details about these scenarios \cite{penfold2011infer,emmert2012statistical,huynh2019gene}. Other scenarios will be thoroughly discussed: whether and how the GRN can be inferred. We also develop novel mathematical methods if necessary. Since such less-studied scenarios lack real data, we implement our novel methods and evaluate on simulated data. For identified regulatory relations, we determine whether they are activation or inhibition. Through this analysis, the existing mathematical methods, our novel methods, and potential methods that will be invented in the future can be treated in the same unified paradigm. This paradigm can be generalized to fit new data types.

There are many subtleties in gene regulation that cannot be fully considered. We need to warn that all inference methods need some (possibly unrealistic) assumptions in mathematics and biology to (over)simplify the problem. Therefore, the inference results should be regarded as informative findings, not ground truth.

This paper is not a review of all important existing GRN inference methods. For example, machine learning methods like the GENIE3 algorithm \cite{huynh2010inferring}, and Bayesian type methods which could easily utilize prior information, are not covered. Our methods are mostly deterministic, meaning that they produce the same inference for fixed data. Besides, this paper mainly focuses on theoretical results, not practical algorithms.

Fig. \ref{diag} describes the dependence relations of all 10 sections and the appendix. Section \ref{biol} briefly introduces possible biological data types regarding GRNs and related biological assumptions. Section \ref{cond} discusses mathematical setup and related mathematical assumptions. Section \ref{main} presents the main results of this paper: with what data, under what assumptions, what we can infer about the GRN structures. Section \ref{concern} discusses biological background about GRN and related measurements, and concerns about related assumptions. Section \ref{existing} contains some classical mathematical results related to GRN structure inference. In Section \ref{novel}, we develop some novel mathematical results that are used in the inference of GRN structures. In Section \ref{imp}, our novel methods are implemented and evaluated on simulated data. Section \ref{scenarios} provides a detailed explanation for the results in Section \ref{main}. We finish with some general discussions in Section \ref{disc}. Appendix \ref{app} contains simulation details for Section \ref{imp}.

This paper is written to draw attention from readers with various backgrounds. Readers who only want to understand the main results could stop at Section \ref{main}. Readers who are not interested in mathematical details could skip Subsections \ref{S31}-\ref{S33}, \ref{ca}, and Sections \ref{existing}-\ref{disc}. Readers who are not interested in biological discussions could skip Section \ref{concern}. 

\begin{figure}
	\center
	$\xymatrix{
		\mathsection2\ar[d]\ar[dr]&\mathsection1\ar[l]\ar[r]\ar[r]&\mathsection3\ar[dll]\ar[d]\ar[ld]\ar[r]&\mathsection6\ar[ldd]\ar[ld]\\
		\mathsection5\ar[rd]&\mathsection4\ar[rd]&\mathsection7\ar[d]\ar[r]&\mathsection8\ar[d]\ar[ld]\\
		&\mathsection{10}&\mathsection9\ar[l]& \text{App}
	}$\\
	\caption{Dependence relations of all 10 sections and the appendix (App) in this paper.}
	\label{diag}
\end{figure}
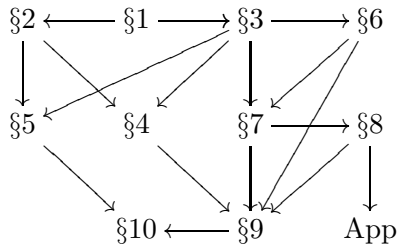

\section{Biological setup}
\label{biol}
\subsection{Experimental data types}
\label{edt}
Consider genes $V_1,\ldots,V_n$. Various types of experimental data can be used to infer GRN structures. They can be classified in four major dimensions: (1) Modality. We can directly measure the expression levels of $V_1,\ldots,V_n$, or only a phenotype $V_0$ that depends on the expression levels of $V_1,\ldots,V_n$. (2) Granularity. We can measure the levels of $V_1,\ldots,V_n$ or $V_0$ for a single cell and repeat many times, so as to obtain a group of random variables $X_1,\ldots,X_n$ or $X_0$ that represent the random values of $V_1,\ldots,V_n$ or $V_0$. We also obtain the corresponding probability distribution of these variables. Alternatively, we can measure these quantities over a large population of cells (bulk level), so that the randomness is averaged out. Then we obtain deterministic results of $X_1,\ldots,X_n$ or $X_0$, denoted by $x_1,\ldots,x_n$ or $x_0$. (3) Intervention. We can measure certain quantities after interfering with the expression of some genes, so that these genes and downstream genes are affected (not fully muted in general). Then we can compare measured quantities before intervention ($x_1,\ldots,x_n$) and quantities after intervention ($x_1',\ldots,x_n'$). We can also observe without any intervention. (4) Temporal resolution. We can measure at a single time point, $x_i(0)$, or measure at multiple time points as a time series, $x_i(0),x_i(1),x_i(2),$ and so on. Two time points are enough for certain inference methods, while some other methods need many more time points. (\textasteriskcentered) When we measure at single-cell level at multiple time points, there is an extra dimension: Can we measure the same cell multiple times, so that we can obtain the joint distribution for multiple time points, $\mathbb{P}[X_1(0)=c_0,X_1(1)=c_1,X_1(2)=c_2]$, or we need to measure different cells at different time points, so that we just have multiple marginal distributions for each time point, $\mathbb{P}[X_1(0)=c_0],\mathbb{P}[X_1(1)=c_1],\mathbb{P}[X_1(2)=c_2]$? This difference has a thermodynamic explanation: without joint distribution, we cannot distinguish equilibrium steady state and non-equilibrium steady state \cite{wang2020mathematical}. 

For these four major dimensions, we have $2^4=16$ different scenarios. Each scenario is named with a combination of four labels: (1) Gene expression or Phenotype; (2) Single-cell or Bulk; (3) Non-interventional or Interventional; (4) One-time or Time series. In four scenarios (Single-cell + Time series), there is an extra dimension of Joint distribution or Marginal distribution, meaning \textbf{a total of 20 scenarios}. For these 20 scenarios, we study how to infer the GRN structure, and the results are summarized in Table \ref{Tab}.

For each dimension, there is one type that is strictly more informative than the other: Gene expression $>$ Phenotype, Single-cell $>$ Bulk, Interventional $>$ Non-interventional, Time series $>$ One-time (the more time points, the more informative), Joint distribution $>$ Marginal distribution. Nevertheless, for more informative data types, generally the experiments are more difficult, more expensive, and less accurate.

There are intermediate states for each dimension. (1) When we can only measure the expression level of one gene (or a few genes), we can regard the measured gene as a phenotype, and apply methods developed for phenotype data. (2) When we increase the number of cells measured in one experiment, due to the central limit theorem \cite{durrett2019probability}, the data variance decreases fast. When the cell number is small enough, so that the stochasticity is significant, we can treat this case the same as the single-cell case. (3) When we can only add interventions on limited genes, some inference methods for interventional data might fail, and we need to treat such data as non-interventional. (4) For time series data, some methods, such as the one in Subsection \ref{GSP}, only require measuring at two time points; some other methods, such as those in Subsection \ref{ODE} and Subsection \ref{partial}, require measuring at many more time points, proportional to the number of genes involved. Therefore, when we have data at a few time points, we might need to use methods for one or two time points.

\subsection{Biological assumptions}
\label{biola}
Gene regulation in reality is complicated. To infer the GRN structure with limited experimental data, we need some biological assumptions.

(A) Besides genes, other non-genetic factors also regulate gene expression (retinoic acid for Fgf8). We assume that such non-genetic factors are kept at constants, so that they do not interfere with determining the regulation between genes.

(B) The regulation between genes might require specific conditions (lactose and glucose for Lacl and \textit{lac} operon). We assume that all such conditions are satisfied, so that we can observe all possible regulations among considered genes.

(C) We assume all genes in the GRN that we want to reconstruct have been considered, and assume that the set of considered genes is not too large. Therefore, the set of genes considered consists of all genes in a small GRN and at most a few irrelevant genes.

(D) If $V_i$ regulates $V_j$, we assume the increase of $V_i$ always leads to the increase (activation) or decrease (inhibition) of $V_j$. We exclude the case that $V_i$ can both activate and inhibit $V_j$ in different situations, although it is not always true \cite{liu2011mechanism}.

With these assumptions or first approximations that are not totally unrealistic, the GRN structure inference problem can be transformed into a well-defined mathematical problem. Nevertheless, these assumptions are not realistic in biology, and they will be further discussed in Subsection \ref{bc}.

\section{Mathematical background and setup}
\label{cond}
\subsection{Introduction of related mathematical assumptions}
\label{S30}
In the following, we introduce some necessary mathematics for understanding the problem of inferring GRN structures. We also introduce four mathematical assumptions that will be used in certain scenarios. Under these assumptions, the underlying GRN is simple enough, or the experimental data are regular enough, so that they follow certain mathematical models. Nevertheless, all assumptions are more or less unrealistic in biology, especially the assumption that the GRN has no directed cycle. We will discuss concerns about these assumptions in Subsection \ref{ca}.

Four assumptions are used in certain scenarios: (1) \textbf{Path Blocking (PB)}: the intervention on one gene has no effect on another gene (or a phenotype), if and only if other intervened genes have already blocked all paths; (2) \textbf{Directed Acyclic Graph (DAG)}: the GRN can be described by a directed graph without cycles; (3) \textbf{Markov and Faithful (MF)}: the distribution of gene expression properly reflects the underlying DAG through conditional independence relations; (4) \textbf{Linear System (LS)}: the gene expression (and possibly phenotype) time series data satisfy a linear ODE system. Most contents of this section, except the path blocking property, are from standard textbooks \cite{G173,PMJ,Murray}. If not interested in mathematical details, readers can remember the names of these assumptions and skip the rest of this section.

\subsection{Graph theory and path blocking property}
\label{S31}
In graph theory, a \textbf{directed graph} consists of several vertices $\{V_1,\ldots,V_n\}$ and some directed edges between different vertices, such as $V_i\to V_j$. 

A GRN can be represented as a directed graph. Each vertex $V_i$ represents a gene, and each directed edge $V_i\to V_j$ means that the expression of gene $V_i$ can regulate the expression of gene $V_j$. We require that the graph is connected. If this graph has several connected components, we can deal with them separately.

If there are edges $V_i\to V_j,V_j\to V_k,\ldots,V_l\to V_m$, then $V_i\to V_j\to V_k\cdots\to V_l\to V_m$ is called a \textbf{path} from $V_i$ to $V_l$. Although the GRN might have cycles, we do not allow a path to pass a vertex more than once. The number of edges in a path is called the \textbf{length} of this path. 

If we can only measure a phenotype $V_0$, then the concerned GRN consists of $V_0$ and genes $V_1,\ldots,V_n$ that have paths leading to $V_0$. We assume no edge starts from $V_0$.

If there is an edge $V_i\to V_j$, then $V_i$ is a \textbf{parent} of $V_j$, and $V_j$ is a \textbf{child} of $V_i$. If there is a path from $V_i$ to $V_j$,  then $V_i$ is an \textbf{ancestor} of $V_j$, and $V_j$ is a \textbf{descendant} of $V_i$. 

Consider any two vertices $V_i,V_j$ and a set $\mathcal{S}\subset\{V_1,\ldots,V_n\}\backslash\{V_i,V_j\}$. If all paths from $V_i$ to $V_j$ need to pass through $\mathcal{S}$, we say that $\mathcal{S}$ \textbf{blocks} $V_i$ to $V_j$. Specifically, if there is no path from $V_i$ to $V_j$, then any set, including the empty set $\emptyset$, blocks $V_i$ to $V_j$.

Assume we can interfere with different sets of genes and measure the expression levels of all genes, so as to compare if they have the same effect on certain genes. The \textbf{path blocking property} means that for any $V_i,V_j$ and any set $\mathcal{S}\subset\{V_1,\ldots,V_n\}\backslash\{V_i,V_j\}$, after interfering with all genes in $\mathcal{S}$, intervention on $V_i$ cannot provide extra influence on the expression of $V_j$, if and only if $\mathcal{S}$ blocks $V_i$ to $V_j$. The path blocking property in this case means a set of genes that blocks $V_i$ to $V_j$ in the directed graph (GRN) also blocks the influence of $V_i$ to $V_j$ in gene expression, and vice versa.  

Assume we can interfere with different sets of genes and measure one phenotype $V_0$, so as to compare if they have the same effect on the measured phenotype. The \textbf{path blocking property} means that for any $V_i$ and any set $\mathcal{S}\subset\{V_1,\ldots,V_n\}\backslash\{V_i\}$, after interfering with all genes in $\mathcal{S}$, intervention on $V_i$ cannot provide extra influence on the expression of $V_0$, if and only if $\mathcal{S}$ blocks $V_i$ to $V_0$. 

\subsection{Directed acyclic graph, {Markov} property, and faithful property}
\label{S32}
A \textbf{directed acyclic graph} (DAG) is a directed graph that has no directed cycles. For example, if we have edges $V_1\to V_2$ and $V_2\to V_3$, then we cannot have edge $V_3\to V_1$, but edge $V_1\to V_3$ is allowed.

In a \textbf{causal DAG} (also called Bayesian network) $\mathcal{G}$ with $n$ vertices, each vertex $V_i$ has an associated random variable $X_i$, which represents the stochastic expression level of gene $V_i$ in a single cell. The joint distribution of these variables is denoted as $\mathbb{P}$. 

%For a variable $X_i$, let $\tilde{X_i}$ denote all the other variables except the parents and descendants of $X_i$. %\textbf{If $\mathbb{P}$ is Markov to $\mathcal{G}$, then $X_i$ and $\tilde{X_i}$ are independent conditioned on $\pi_{X_i}$.}

$\mathbb{P}$ is said to be \textbf{Markov} to $\mathcal{G}$, if $\mathbb{P}(X_1,\ldots,X_n)=\prod^{n}_{i=1}\mathbb{P}(X_i\mid \pi_{X_i})$, where $\pi_{X_i}$ means the expression levels of $V_i$'s parents. If all the conditional independence relations in $\mathbb{P}$ also appear in any other $\mathbb{P}'$ that is Markov to $\mathcal{G}$, then we say that $\mathbb{P}$ is \textbf{faithful} to $\mathcal{G}$. If a distribution $\mathbb{P}$ is Markov and faithful to $\mathcal{G}$, then the structure of $\mathcal{G}$ properly reflects the causal relations in $\mathbb{P}$.

Consider a DAG $\mathcal{G}$. A common model that fits this DAG is $X_i=f_i(\pi_{X_i})+\epsilon_i$. This means the value of $V_i$ depends on the values of its parents, plus an independent noise. In general, the joint distribution generated by this model is Markov and faithful to $\mathcal{G}$.

\subsection{Linear system}
\label{S33}
If we measure the bulk level gene expression time series, the expression level of a gene $V_i$ can be regarded as a variable $x_i(t)$ that changes continuously along time. If we measure the single-cell level gene expression $X_i$, then its expectation is $x_i$. The expression levels $\{x_i\}$ of all genes along time satisfy a \textbf{linear system}, or have the \textbf{linearity} property, if the following linear ordinary differential equations (ODEs) hold.
\begin{equation}
	\label{eq1}
	\begin{split}
		\mathrm{d}x_1/\mathrm{d}t &=a_{11}x_1+a_{12}x_2+\cdots+a_{1n}x_n+b_1,\\
		&\cdots\\
		\mathrm{d}x_n/\mathrm{d}t &=a_{n1}x_1+a_{n2}x_2+\cdots+a_{nn}x_n+b_n.
	\end{split}
\end{equation}
This system can be rewritten as $\mathrm{d}\vec{x}/\mathrm{d}t=A\vec{x}+\vec{b}$, where $\vec{x}=(x_1,\ldots,x_n)'$, $A=[a_{ij}]$, $\vec{b}=(b_1,\ldots,b_n)'$. Here $b_i$ describes the base synthesis rate of gene $V_i$, $a_{ii}$ describes the degradation of $V_i$, and $a_{ij}$ describes the effect of $V_j$ on the expression of $V_i$. If $a_{ij}>0$, $V_j$ activates $V_i$; if $a_{ij}<0$, $V_j$ inhibits $V_i$; if $a_{ij}=0$, $V_j$ does not affect $V_i$. If we want the solution of Eq. \ref{eq1} to converge to a stable fixed point $\vec{x}_0$ with positive components, one approach is to check the condition for sign stability \cite{jeffries1977matrix,bone1988qualitative}. Another approach is to assume $b_i>0$, $a_{ij}\ge 0$ for $i\ne j$, $\sum_{j=1}^na_{ij}<0$. Then the Perron-Frobenius theorem states that all eigenvalues of $A$ have negative real parts, so that the fixed point $\vec{x}_0=-A^{-1}\vec{b}$ is stable and has positive components.

\section{Inferring {GRN} structure with different data types}
\label{main}
We have discussed different data types and different assumptions in Section \ref{biol} and Section \ref{cond}. In this section, we present the main results of this paper: with what data type, under what assumptions, what we can infer about the GRN structure. See Table \ref{Tab} for a summary of all the results. For example, in Scenario 6 (Gene, Bulk, One-time, Interventional), we can fully infer the GRN structure under the path blocking assumption, and we can partially infer the GRN structure under the directed acyclic graph assumption. These results are explained in Section \ref{scenarios}, based on the mathematical methods in Section \ref{existing} and Section \ref{novel}.

\begin{table}[]
	\begin{tabular}{|l|l|l|l|l|l|}
		\hline
		\multicolumn{2}{|l|}{\multirow{2}{*}{}}                                                                                             & \multicolumn{2}{l|}{One-Time}                                                                                                                                                                   & \multicolumn{2}{l|}{Time Series}                                                                                                                                                                                                                                                             \\ \cline{3-6} 
		\multicolumn{2}{|l|}{}                                                                                                              & \begin{tabular}[c]{@{}l@{}}Non-\\ Intervention\end{tabular}                                & Intervention                                                                                       & \begin{tabular}[c]{@{}l@{}}Non-\\ Intervention\end{tabular}                                                                   & Intervention                                                                                                                                                 \\ \hline
		\multirow{2}{*}{\begin{tabular}[c]{@{}l@{}}Gene\\ Expre\\-ssion\end{tabular}} & \begin{tabular}[c]{@{}l@{}}Single-\\ Cell\end{tabular} & \begin{tabular}[c]{@{}l@{}}Scenario 1:\\ \\ MF+DAG: \\ partial.\\ \end{tabular} & \begin{tabular}[c]{@{}l@{}}Scenario 2:\\ \\ PB: full.\\ DAG: partial.\\ MF+DAG: \\ full.\\ \end{tabular} & \begin{tabular}[c]{@{}l@{}}Scenario 3 a/b:\\ \\ 3a Joint:\\ UC: full.\\  3b Marginal:\\ MF+DAG: \\ partial.\\ \end{tabular} & \begin{tabular}[c]{@{}l@{}}Scenario 4 a/b:\\ \\ 4a Joint:\\ UC: full.\\  4b Marginal:\\ LS: full.\\ PB: full.\\ DAG: partial.\\ MF+DAG: full.\\ \end{tabular} \\ \cline{2-6} 
		& Bulk                                                   & \begin{tabular}[c]{@{}l@{}}Scenario 5:\\ \\ No.\end{tabular}                                 & \begin{tabular}[c]{@{}l@{}}Scenario 6:\\ \\ PB: full.\\ DAG: partial.\end{tabular}                    & \begin{tabular}[c]{@{}l@{}}Scenario 7:\\ \\ No.\end{tabular}                                                                    & \begin{tabular}[c]{@{}l@{}}Scenario 8:\\ \\ LS: full.\\ PB: full.\\ DAG: partial.\end{tabular}                                                                   \\ \hline
		\multirow{2}{*}{\begin{tabular}[c]{@{}l@{}}Pheno\\ -type\end{tabular}}                                                 & \begin{tabular}[c]{@{}l@{}}Single-\\ Cell\end{tabular} & \begin{tabular}[c]{@{}l@{}}Scenario 9:\\ \\ No.\end{tabular}                                 & \begin{tabular}[c]{@{}l@{}}Scenario 10:\\ \\ PB: partial.\end{tabular}                               & \begin{tabular}[c]{@{}l@{}}Scenario 11 a/b:\\ \\ No.\end{tabular}                                                                   & \begin{tabular}[c]{@{}l@{}}Scenario 12 a/b:\\ \\ PB: partial.\\ LS+DAG: \\ partial*.\\ PB+LS+DAG: \\ partial*.\end{tabular}                                                                    \\ \cline{2-6} 
		& Bulk                                                   & \begin{tabular}[c]{@{}l@{}}Scenario 13:\\ \\ No\end{tabular}                                & \begin{tabular}[c]{@{}l@{}}Scenario 14:\\ \\ PB: partial.\end{tabular}                               & \begin{tabular}[c]{@{}l@{}}Scenario 15:\\ \\ No.\end{tabular}                                                                   & \begin{tabular}[c]{@{}l@{}}Scenario 16:\\ \\ PB: partial.\\ LS+DAG: \\ partial*.\\ PB+LS+DAG: \\ partial*.\end{tabular}                                                                    \\ \hline
	\end{tabular}
	\caption{GRN structure inference with different data types: under what assumptions, what structures can be inferred. There are 16 scenarios classified by the following dimensions of data types: Gene Expression vs. Phenotype; Single-Cell vs. Bulk; One-Time vs. Time Series; Non-Interventional vs. Interventional. In Scenarios 3/4/11/12, there is an extra dimension of Joint vs. Marginal. There are different assumptions: PB: path blocking; DAG: directed acyclic graph; MF: Markov and faithful; LS: linear system; UC: unconditional. Full/partial/no means all/some/no GRN structures can be inferred. For example, ``MF+DAG: partial'' means under MF assumption and DAG assumption, GRN structure can be partially inferred. The asterisk * in Scenarios 12/16 means for some identified regulatory relations, we cannot determine whether they are activation or inhibition.}
	\label{Tab}
\end{table}

The basic method of establishing causal relation is to vary one variable and examine if the other variable changes. When we cannot interfere with gene expression, the only chance is to profile gene expression on single-cell level. Within a single cell, the stochastic fluctuation of gene expression plays a similar role as intervention. Nevertheless, by comparing Scenario 1 vs. Scenario 2 and comparing Scenario 3 vs. Scenario 4, we can see that the stochastic fluctuation is less informative than intervention. 

With interventional phenotype data (Scenarios 10, 12, 14, 16), the GRN is analogous to a black box. Although the gene expression levels are unknown, we can interfere by manipulating certain genes and use the resulting phenotype as output. Under different assumptions, we can partially reconstruct the GRN structure. 

We can evaluate the performance of each scenario based on the level of GRN structures that can be inferred. Fully recovered GRN means 2 points, and partially recovered GRN means 1 point. For example, Scenario 1 (partial GRN) has 1 point. The summation over eight scenarios (1-8) with gene expression data is 11 points, and the summation over eight scenarios with phenotype data (9-16) is 4 points. We compare the overall score for each data dimension. Gene expression vs. Phenotype: 11 vs. 4; Single-cell vs. Bulk: 9 vs. 6; Time series vs. One-time: 8 vs. 7; Interventional vs. Non-interventional: 12 vs. 3. We can see that the gene expression data and interventional data are more informative.

Scenario 4 (Gene, Single-cell, Time series, Interventional) is the most informative case. We can see that in Scenario 4, the GRN structure can be fully inferred under various assumptions.

\section{Biological details and concerns}
\label{concern}

\subsection{Measurements related to {GRN}}
\label{datatype}
We briefly introduce biological experiments related to inferring GRN structures and their restrictions, as an extension for Subsection \ref{edt}. We do not aim at covering all important related papers. Readers may refer to some review papers for detailed summaries on related biological techniques \cite{ding2020systematic,wang2019comparative,perrimon2010vivo}. 

The expression of a protein-coding gene consists of transcription (DNA to mRNA) and translation (mRNA to protein). Therefore, to measure the expression level of a given gene, we can either measure the amount of the corresponding mRNA or measure the amount of the corresponding protein. There are many methods with various reliabilities and accuracies to measure the amount of mRNA and/or protein on the bulk level and the single-cell level \cite{andrews2018identifying,wang2009rna,sinkala2017profiling}. By now (2021), these methods are not 100\% accurate, especially on single-cell level \cite{svensson2020droplet}. A more important problem is that most mRNAs and proteins are confined within living cells (except for secreted proteins). This means cells have to be killed before these molecules can be harvested, and then some analytical methods can be applied to quantify the abundance of these molecules. Therefore, with these methods, a cell or a cell population can be measured only once \cite{andrews2018identifying}.

Many mRNAs have less than 20 copies in one cell \cite{so2011general}. Thus the gene expression of a single cell is too stochastic to be described by a deterministic model. We can repeat the gene expression measurement experiment over multiple single cells and obtain a probability distribution of single-cell gene expression \cite{kaern2005stochasticity}. When the observation is based on a large number of cells (bulk level), stochasticity is averaged out, so that the dynamics should be deterministic. Theoretically, repeating the gene expression measurement experiment over a large number of cells should produce the same result. In reality, various technical issues weaken the repeatability of measuring gene expression \cite{rondina2020longitudinal}. For example, external factors that affect gene expression, such as nutrition concentrations, are hardly the same for two cell populations. Therefore, when we sample two populations from the same cell line (especially if we sample at different time points), their gene expression profiles might not be exactly the same.

Traditional approaches in the studies of gene expression are observational, meaning that we directly measure cell(s) at stationary without interfering with gene expression. Recent techniques like gene knockdown or gene knockout allow us to temporarily or permanently alter the expression of genes of interest \cite{nikam2018journey,pickar2019next,gujral2014exploiting}. Most interventions decrease the expression of corresponding genes, but increasing is also possible \cite{wang2013integration}. After the intervention, we can measure the gene expression data at a certain time point. By now (2021), on bulk level, commonly used knockdown interventions cannot robustly decrease the expression of targeted genes to zero \cite{munkacsy2016validation}. For example, Hurley et al. applied siRNA to disrupt certain genes. For only around 70\% genes, the bulk level expression decreases to less than 40\% of the original level. For around 10\% genes, siRNAs even increase their bulk expression levels \cite{hurley2012gene}. On the single-cell level, we can measure and select cells that are successfully intervened. However, the measurement is not 100\% accurate, so that the selection is not always correct. Besides, selection might be biased towards cells with low baseline expression levels of targeted genes. Interventions cannot be used to maintain gene expression to fixed non-zero levels.

After the intervention, we hope to measure how the gene expression evolves over time. Under current technologies, we can only measure a cell or a cell population once. Therefore, at the bulk level, we can measure different populations at different time points after the same intervention. The gene expression level at the bulk level is a deterministic number. Thus theoretically, measuring different populations at different time points is equivalent to measuring the same population at different time points. At the single-cell level, we can repeat the measurement for different cells at the same time point to obtain the marginal distribution of gene expression at each time point. However, since we can only measure a cell once, we cannot obtain the joint distribution of gene expression at different time points. The joint distribution can provide extra information (correlation coefficient, etc.) than marginal distributions. Regardless of technology restrictions, we discuss both cases, depending on whether the joint distribution can be measured.

Some phenotypes, such as growth rate, number of molecules released, and drug resistance, could be used to reflect gene expression levels \cite{galluzzi2012prognostic}. These phenotypes can be measured on bulk level or single-cell level. Besides, some phenotypes can be measured without perturbing or even killing the cells. This means such phenotypes can be measured at different time points for the same cell(s). Nevertheless, phenotypic transitions involve complicated phenomena and mechanisms \cite{jiang2017phenotypic}. Therefore, we need to be cautious when utilizing phenotypes in determining GRN structures.

\subsection{Concerns about biological {GRN} assumptions}
\label{bc}
In Subsection \ref{biola}, we propose four biological assumptions to simplify the GRN structure inference problem. Here we discuss biological concerns about these assumptions.

(A) Various non-genetic factors, such as vitamins and minerals, directly regulate gene expression. Since this paper focuses on regulations between genes, we assume such non-genetic factors are kept at constants. Nevertheless, we do not fully understand what factors can affect gene expression, and it is difficult to control the abundance of so many kinds of molecules. Besides, when we manipulate the expression of genes, the concentrations of some non-genetic factors might also be changed \cite{reidling2008mechanisms}.

(B) The regulation between genes might occur only under specific conditions. We assume that such conditions are satisfied, so that all regulations can be observed. Since we need to infer whether a regulation between genes exists, we cannot guarantee that the conditions for this regulation are satisfied. We can also remove this assumption, and just aim at reconstructing the GRN that operates under naive conditions.

(C) To correctly infer the GRN structure, we need to assume that all genes in the target GRN are considered. For example, if the true GRN is $V_1\to V_2\to V_3$, but $V_2$ cannot be measured or intervened with, then we can observe that varying $V_1$ in any case would lead to the change of $V_3$. Therefore, we would obtain a false direct edge $V_1\to V_3$. We also need to assume that the set of considered genes is not too large. The amount of data needed and the computational time cost for GRN structure inference methods grow fast with the number of genes considered. If we cannot exclude most genes that are irrelevant to the target GRN beforehand, then either the inference methods behave poorly, or the cost (time and money) is unbearable. Nevertheless, before inferring the GRN structure, it is almost impossible to accurately distinguish which genes are involved or not involved in the target GRN.

(D) To simplify the classification of activation and inhibition, we assume the regulation is monotonic, meaning that if $V_i$ regulates $V_j$, then the increase of $V_i$ always leads to the increase or decrease of $V_j$. In reality, Crn1 at different concentrations can either activate or inhibit the Arp2/3 complex in yeast \cite{liu2011mechanism}. This illustrates that gene regulation might be non-monotonic.

\subsection{Concerns about mathematical {GRN} assumptions}
\label{ca}
In Section \ref{cond}, we propose four mathematical assumptions that will be used in certain scenarios. Here we discuss biological concerns about these assumptions.

(1) The Path Blocking (PB) assumption states that the intervention on one gene has no effect on another gene (or a phenotype), if and only if other intervened genes block all paths. Consider a GRN $V_1\to V_2\to V_3$. Some interventions (e.g., siRNA) are not strong enough, so that after interfering with $V_2$, $V_1$ can still affect $V_2$, then $V_3$. On the other hand, if we assume that $V_1$ is necessary for the expression of $V_2$, then as we knock out $V_1$, $V_2$ cannot express. In this case, interventions on $V_2$ have no extra effect on $V_3$, although the path $V_2\to V_3$ is not blocked. Even though the path blocking assumption is powerful in the GRN structure inference, it is not very realistic.

(2) The directed acyclic graph (DAG) assumption states that the GRN has no cycle, which is crucial in many scenarios, but it is truly unrealistic in biology. Introducing four factors can induce somatic cells to become pluripotent stem cells, which have very different gene expression patterns \cite{takahashi2006induction}. Therefore, the related GRN is extremely large, and it is unavoidable that some feedback loops are contained \cite{bansal2006inference}. Although some reported GRNs have no cycles, it is possible that some unobserved genes form cycles in the real GRN.

(3) The Markov and faithful (MF) assumption states that the conditional independence relations in the distribution of gene expression are consistent with the causal DAG. Under this assumption, conditional independence should be correctly identified. However, due to the heterogeneity of cells \cite{wang2018some}, there might be multiple cell types, each of which has its own gene expression distribution, and the measured gene expression distribution might be the combination of multiple distributions. Such mixing might make independent genes show pseudo dependence. Mathematically, assume $X_1,X_2$ (expression of two genes for cell type 1) are independent, and $X_1',X_2'$ (expression of two genes for cell type 2) are independent. Define $\bar{X}_1=X_1$ with probability $0.5$, and $\bar{X}_1=X_1'$ with probability $0.5$, and $\bar{X}_2$ is defined similarly (observed expression of two genes for mixed cell population). Then $\bar{X}_1$ and $\bar{X}_2$ might not be independent.

(4) The linear system (LS) assumption states that the net change rate of one gene (corresponding mRNA or protein) linearly depends on the quantity of another gene, regardless of other genes. In reality, the regulation of gene expression might be highly nonlinear. Specifically, if gene $V_1$ inhibits gene $V_2$, then when the level of $V_2$ is $0$, the net change rate of $V_2$ might still be negative, a contradiction. One alternative approach is to add nonlinear terms to the dynamics \cite{polynikis2009comparing}. We need to determine the forms of such nonlinear terms in advance, and such terms should depend on finitely many parameters. Although this approach can better fit the real dynamics, there are some problems: (i) we do not know whether the forms of nonlinear terms are correct; (ii) with more parameters to infer, the reliability decreases; (iii) some methods, such as the one in Subsection \ref{partial}, would fail with nonlinear terms.

\section{Existing related mathematical results}
\label{existing}
In this section, each subsection contains a known mathematical method that will be used to infer GRN structures in certain scenarios. At the end of each subsection, there is an ``A/I Remark'' that explains how to determine whether an identified regulatory relation is activation or inhibition.

\subsection{Use conditional independence to infer {DAG} structure}
\label{DAG}
This subsection is a standard topic in causal inference. Readers may refer to related monographs for background and details \cite{PMJ}.

For a distribution $\mathbb{P}$ that is Markov and faithful to a causal DAG, there is an edge between $V_i$ and $V_j$ if and only if $V_i$ and $V_j$ are independent conditioned on some other variables. (When some genes are highly correlated, a more rigorous method is to detect a Markov boundary and determine its uniqueness \cite{wang2020causal}.) Consider three DAGs: $V_1\to V_2\to V_3$, $V_1\leftarrow V_2 \leftarrow V_3$, $V_1\leftarrow V_2 \to V_3$. If $\mathbb{P}$ is Markov and faithful to one of these DAGs, then it is also Markov and faithful to the other two. Thus we cannot distinguish them. Therefore, using conditional independence, we can determine all edges of the unknown DAG, but some edge directions are unknown.

\begin{remark} 
	For an inferred edge $V_i\to V_j$, we can calculate the conditional covariance between $V_i$ and $V_j$, conditioned on other possible parents of $V_j$. Positive conditional covariance means activation, and negative conditional covariance means inhibition. 
\end{remark}

\subsection{Inferring {GRN} structure with {ODE} model}
\label{ODE}
For the linear ODE system Eq. \ref{eq1}, if we know the value of $x_i(t)$ at different time points, we can calculate $\mathrm{d}x_i/\mathrm{d}t$, and then calculate the parameters $\{a_{ij}\}$ and $\{b_i\}$ by solving a linear algebraic equation system \cite{pollicott2012extracting}. Here we need to assume that the interventions added on different genes only change the initial values $x_i(0)$, but not system parameters $\{a_{ij}\}$ and $\{b_i\}$. If the genes in a GRN satisfy such a linear system, then $a_{ij}\ne 0$ if and only if there is an edge $V_j\to V_i$. We can also assume the ODE system has nonlinear terms with known forms \cite{polynikis2009comparing}. Similarly, we can solve the parameters and determine the regulatory relationships. One can even test different nonlinear systems and select the best one \cite{bansal2006inference}.

\begin{remark}
	In the linear ODE model, if $a_{ij}>0$, $V_j$ activates $V_i$; if $a_{ij}<0$, $V_j$ inhibits $V_i$. The value of $a_{ij}$ quantifies the regulation strength.
\end{remark}

\subsection{Determining {GRN} in stochastic process model}
\label{GSP}
Assume we can measure the same cell at different time points, then we obtain the joint distribution of gene expression profile as a stochastic process. This model is also called the dynamic Bayesian network. For example, we consider three variables at different time points, $[X_1(0),X_2(0),X_3(0)]$, $[X_1(1),X_2(1),X_3(1)],\ldots$. Here $X_1(t),X_2(t),X_3(t)$ are the expression levels of genes $V_1,V_2,V_3$ at time $t$. We assume there are no other genes that affect the expression of $V_1,V_2,V_3$. 

Assume the underlying GRN is $V_1\to V_2\to V_3$, then this means current $X_1$ can directly affect future $X_1$, current $X_1,X_2$ can directly affect future $X_2$, and current $X_2,X_3$ can directly affect future $X_3$. The causal relations between different variables in this process are shown in Fig. \ref{F1}. An arrow can only point from a variable at time $t$ to a variable at time $t+1$. Therefore, this illustration is a DAG. We can project this process DAG along the time axis to obtain the underlying GRN.  

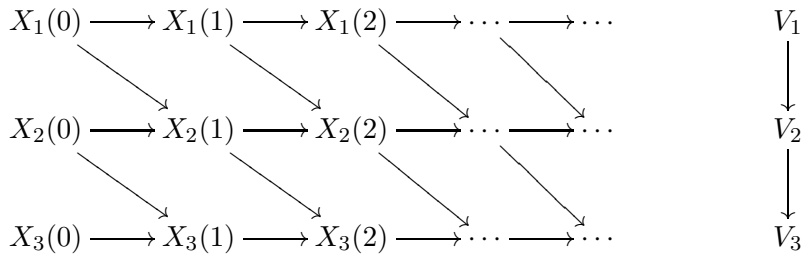
\begin{figure}
	\center
	$\xymatrix{
		X_1(0)\ar[r]\ar[rd]&X_1(1)\ar[r]\ar[rd]&X_1(2)\ar[r]\ar[rd]&\cdots\ar[r]\ar[rd]&\cdots&&V_1\ar[d]\\
		X_2(0)\ar[r]\ar[rd]&X_2(1)\ar[r]\ar[rd]&X_2(2)\ar[r]\ar[rd]&\cdots\ar[r]\ar[rd]&\cdots&&V_2\ar[d]\\
		X_3(0)\ar[r]&X_3(1)\ar[r]&X_3(2)\ar[r]&\cdots\ar[r]&\cdots&&V_3
	}$\\
	\caption{Stochastic process model. $X_1(0)$ can affect $X_1(1)$ and $X_2(1)$; $X_2(0)$ can affect $X_2(1)$ and $X_3(1)$; $X_3(0)$ can affect $X_3(1)$. The right most $V_1\to V_2\to V_3$ is the corresponding GRN.}
	\label{F1}
\end{figure}

If we have the joint distribution of two neighboring time points, then we can use conditional independence to infer the structure of the process DAG. Since each arrow can only point from an earlier time to a later time, the direction is determined. The projection of the process DAG is the underlying GRN. Since $X_i(t+1)$ naturally depends on $X_i(t)$, we do not project this edge as $V_i\to V_i$. If the underlying GRN has cycles, the process DAG is still a DAG. Thus the GRN with cycles can also be inferred.

If the time step of this stochastic process is not small enough, the inference might be problematic. For example, in Fig. \ref{F1}, if the observation time step is 2, so that $[X_1(1),X_2(1),X_3(1)]$ are invisible, then an indirect relation $X_1(0)\to X_2(1)\to X_3(2)$ will be misinterpreted as a direct edge $X_1(0)\to X_3(2)$, meaning that the inferred GRN has an extra edge $X_1\to X_3$.

For GRN $V_1\to V_2\to V_3$, from Fig. \ref{F1}, we can see that conditioned on $[X_2(0),X_2(1),\ldots]$, $[X_1(0),X_1(1),\ldots]$ and $[X_3(0),X_3(1),\ldots]$ are independent. However, for any $t$, even if the process has reached stationary, we might not have $X_1(t)\ind X_3(t)\mid X_2(t)$. Therefore, in the stochastic process model, if we only have the stationary distribution, we cannot use conditional independence to infer the GRN structure. The causal DAG model discussed in Subsection \ref{DAG} might not be regarded as the stationary situation of the stochastic process model.

\begin{remark}
	Similar to the discussion in Subsection \ref{DAG}, we can use conditional covariance to determine whether an arrow represents activation or inhibition. 
\end{remark}

\section{Novel related mathematical results}
\label{novel}
In this section, each subsection contains a novel mathematical method that will be used to infer GRN structures in certain scenarios. At the end of each subsection, there is an ``A/I Remark'' that explains how to determine whether an identified regulatory relation is activation or inhibition. In Section \ref{imp}, these methods are implemented and evaluated on simulated data.

\subsection{Use path blocking relation to infer {GRN} structure with gene expression data}
\label{PBG}
Consider an unknown GRN with genes $V_1,\ldots,V_n$. We can intervene with any genes and measure the expression profile of all genes. Assume the path blocking property holds. Then for each $V_i,V_j$ and each subset $\mathcal{S}$ of $\{V_1,\ldots,V_n\}\backslash\{V_i,V_j\}$, we know whether $\mathcal{S}$ blocks $V_i$ to $V_j$, meaning that any directed path from $V_i$ to $V_j$ passes at least one vertex in $\mathcal{S}$. In this case, we can infer the whole GRN: there is an edge $V_i\to V_j$ if and only if other vertices cannot block $V_i$ to $V_j$. 

See Subsection \ref{80} and Appendix \ref{apppbg} for the performance of this method on simulated data and corresponding discussions. 

\begin{remark}
	For an identified edge $V_i\to V_j$, add interventions on certain genes, so that other paths from $V_i$ to $V_j$ are blocked. Then $V_i$ activates $V_j$ if the decrease of $V_i$ leads to the decrease of $V_j$, and $V_i$ inhibits $V_j$ if the decrease of $V_i$ leads to the increase of $V_j$.
\end{remark}

\subsection{Use path blocking relation to infer {GRN} structure with phenotype data}
\label{PBP}
Consider an unknown GRN with genes $V_1,\ldots,V_n$ and a phenotype $V_0$. We can intervene with any genes, but we can only measure the level of the phenotype, not the genes. We assume there is no edge starting from $V_0$. Assume the path blocking property holds. Then for each $V_i$ and each subset $\mathcal{S}$ of $\{V_1,\ldots,V_n\}\backslash\{V_i\}$, we know whether $\mathcal{S}$ blocks $V_i$ to $V_0$, meaning that any directed path from $V_i$ to $V_0$ passes at least one vertex in $\mathcal{S}$. In this case, the GRN structure can be partially inferred (at least $n$ edges if there is a directed path from each $V_i$ to $V_0$).

If a subset $\mathcal{S}$ of $\{V_1,\ldots,V_n\}\backslash\{V_i\}$ blocks $V_i$ to $V_0$, but any proper subset of $\mathcal{S}$ cannot block $V_i$ to $V_0$, then $\mathcal{S}$ is called a minimal blocking set. If a blocking set $\mathcal{S}$ is not minimal, then $\mathcal{S}$ contains a blocking subset that is minimal. Define $\beta(V_i)$ to be all minimal subsets that block $V_i$ to $V_0$. If two GRNs have the same path blocking relations, then they have the same $\beta(V_i)$ for each $V_i$, and vice versa.

For the two GRNs in Fig. \ref{expb}, $\beta(V_1)=\{\{V_3\}\}$, $\beta(V_2)=\{\{V_3\}\}$, $\beta(V_3)=\emptyset$. Thus they are equivalent in the sense of path blocking. If a directed edge appears in all equivalent GRNs, we can determine that this edge exists; if a directed edge appears in none of equivalent GRNs, we can determine that this edge does not exist; if a directed edge appears in some but not all equivalent GRNs, we cannot determine whether this edge exists.

\begin{figure}
	\center
	$\xymatrix{
		V_1\ar[r]&V_3\ar[r]&V_0&&V_1\ar[r]\ar[d]&V_3\ar[r]&V_0\\
		V_2\ar[ru]&&&&V_2\ar[ru]&&
	}$\\
	\caption{Two equivalent GRNs in the sense of path blocking. The regulation path from $V_1$ to $V_0$ (or $V_2$ to $V_0$) is fully blocked if and only if we add intervention on $V_3$.}
	\label{expb}
\end{figure}
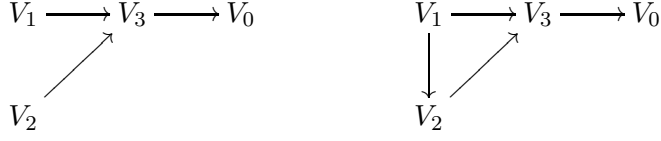

\begin{proposition}
	\label{pbpt}
	The following procedure describes how to determine certain edges. (1) There is an edge $V_i\to V_0$ if and only if $\beta(V_i)=\emptyset$. (2) If there exists $\mathcal{S}\in\beta(V_i)$, so that $V_j\notin \mathcal{S}$, and $\mathcal{S}$ cannot block $V_j$ to $V_0$, then there is no edge $V_i\to V_j$. (3) If (2) is not satisfied, but there exists $\mathcal{S}\in\beta(V_i)$, so that $V_j\in \mathcal{S}$, then there is an edge $V_i\to V_j$. (4) If $\beta(V_i)=\emptyset$, or for any $\mathcal{S}\in\beta(V_i)$, we have $V_j\notin \mathcal{S}$, and $\mathcal{S}$ blocks $V_j$ to $V_0$, then we cannot determine whether $V_i\to V_j$ exists.
\end{proposition}

\begin{proof}
	Assume $\emptyset\in \beta(V_i)$ (notice this is not $\beta(V_i)=\emptyset$), then there is no path from $V_i$ to $V_0$. Obviously, there is no edge $V_i\to V_0$. Besides, if $\emptyset\notin \beta(V_j)$, meaning that there is a path from $V_j$ to $V_0$, then we have case (2), and there is no edge $V_i\to V_j$. If $\emptyset\in \beta(V_j)$, meaning that there is no path from $V_j$ to $V_0$, then we have case (4). Now $V_i,V_j$ are isolated from $V_0$, whether edge $V_i\to V_j$ exists does not affect the blocking relations.
	
	In the rest of this proof, assume $\emptyset\notin \beta(V_i)$.
	
	(1) If there is an edge $V_i\to V_0$, then other vertices cannot block this path. If there is no edge $V_i\to V_0$, then $\{V_1,\ldots,V_n\}\backslash\{V_i\}$ can block $V_i$ to $V_0$, and it contains a minimal blocking subset.
	
	(2) Choose a path $V_j\to\cdots\to V_0$ that cannot be blocked by $\mathcal{S}$. If there exists an edge $V_i\to V_j$, then the path $V_i\to V_j\to\cdots\to V_0$ cannot be blocked by $\mathcal{S}$, a contradiction.
	
	(3) Define $\mathcal{S}'=\mathcal{S}\backslash\{V_j\}$. Assume there is no edge $V_i\to V_j$. Since $\mathcal{S}'$ cannot block $V_i$ to $V_0$, there is a path $V_i\to \cdots \to V_k\to V_j$ that is not blocked by $\mathcal{S}'$, and there is a path $V_j\to \cdots \to V_0$ that is not blocked by $\mathcal{S}'$. 
	
	$V_k$ cannot appear in the path $V_j\to \cdots \to V_0$. Otherwise, we can combine these two paths and shrink out $V_j$, to obtain a path $V_i\to \cdots\to V_k\to \cdots V_0$. This path from $V_i$ to $V_0$ does not pass through $V_j$, and it cannot be blocked by $\mathcal{S}'$, which means it cannot be blocked by $\mathcal{S}$, a contradiction.
	
	Find all possible $V_l$ with the property that there exists a path $V_i\to \cdots \to V_l\to V_j$ that is not blocked by $\mathcal{S}'$, then $V_l$ also cannot appear in the path $V_j\to \cdots \to V_0$. Define the set of all such $V_l$ (including $V_k$) to be $\mathcal{T}$. $\mathcal{T}\cup\mathcal{S}'$ blocks any path from $V_i$ to $V_0$, and it cannot block the path $V_j\to \cdots \to V_0$. We can find a minimal blocking subset in $\mathcal{T}\cup\mathcal{S}'$, which belongs to $\beta(V_i)$, and it cannot block $V_j$ to $V_0$. Thus the condition of (2) is satisfied, a contradiction.
	
	(4) If the condition of (4) is satisfied, we show that when the edge $V_i\to V_j$ does not exist, adding this edge into the GRN does not change path blocking relations, and when the edge $V_i\to V_j$ exists, deleting this edge from the GRN does not change path blocking relations. Therefore, we cannot determine whether the edge $V_i\to V_j$ exists.
	
	If $\beta(V_i)=\emptyset$, then there is an edge $V_i\to V_0$. Adding or deleting $V_i\to V_j$ does not change $\beta(V_i)$. If $\mathcal{S}$ blocks $V_k$ to $V_0$, then $\mathcal{S}$ also blocks $V_k$ to $V_i$, and adding or deleting $V_i\to V_j$ does not matter. If $\mathcal{S}$ does not block $V_k$ to $V_0$, then depending on whether $\mathcal{S}$ can block $V_k$ to $V_i$ or not, we can both directly see that the existence of $V_i\to V_j$ does not matter.
	
	In the following, assume $\beta(V_i)\ne \emptyset$. Assume the edge $V_i\to V_j$ is not in the GRN. Consider $\mathcal{S}$ that blocks $V_i$ to $V_0$, and $\mathcal{R}$ that cannot block $V_i$ to $V_0$. Since $\mathcal{S}$ contains a minimal subset, and this subset blocks $V_j$ to $V_0$, after adding the edge $V_i\to V_j$, it can still block $V_i$ to $V_0$. After adding the edge $V_i\to V_j$ in the GRN, $\mathcal{R}$ still cannot block $V_i$ to $V_0$. For $V_k$ that has a path to $V_i$, we can use the same argument to show that blocking subsets of $V_k$ to $V_0$ are not changed.  
	
	Assume the edge $V_i\to V_j$ is in the GRN. We prove that after deleting this edge, the condition of (4) still holds. Then using the above argument, we can see that after deleting this edge, adding it back does not change the path blocking relations. After deleting this edge, if there is a minimal subset $\mathcal{T}$ that contains $V_j$ and blocks $V_i$ to $V_0$, then before deleting the edge $V_i\to V_j$, $\mathcal{T}$ is still a minimal subset that blocks $V_i$ to $V_0$, contradicting to the condition of (4). After deleting this edge, assume there is a minimal subset $\mathcal{R}$ that blocks $V_i$ to $V_0$ but does not contain $V_j$ and does not block $V_j$ to $V_0$. Then before deleting the edge $V_i\to V_j$, $\mathcal{R}\cup\{V_j\}$ blocks $V_i$ to $V_0$, and it contains a minimal blocking subset $\mathcal{R}'$ that contains $V_j$. The reason for $V_j\in\mathcal{R}'$ is that $\mathcal{R}$ cannot block $V_j$ to $V_0$ after deleting the edge $V_i\to V_j$. Such $\mathcal{R}'$ violates the condition of (4). In sum, after deleting the edge $V_i\to V_j$, $\beta(V_i)$ is not changed. Similarly, for another $V_k$, $\beta(V_k)$ is not changed. 
	
\end{proof}

Without restrictions on connectivity, if there is no edge in this network, then we cannot identify any edges. In the rest of this subsection, assume there is a directed path from each $V_i$ to $V_0$.

Consider the shortest path from $V_i$ to $V_0$, $V_i\to V_k \to \cdots\to V_0$. Define $\mathcal{S}$ to be the set of all children vertices of $V_i$, then $\mathcal{S}$ blocks $V_i$ to $V_0$. Since vertices in $\mathcal{S}\backslash \{V_k\}$ cannot be closer to $V_0$ than $V_k$, $\mathcal{S}\backslash \{V_k\}$ cannot block the path $V_i\to V_k \to \cdots\to V_0$. Therefore, a minimal subset of $\mathcal{S}$ that blocks $V_i$ to $V_0$ must contain $V_k$. This means that the edge $V_i\to V_k$ can be identified. Therefore, we can identify at least one edge starting from each $V_i$. For an unknown GRN, we can identify at least $n$ edges. Consider the GRN $V_n\to\cdots\to V_1\to V_0$. We can identify exactly $n$ edges in this GRN. Thus $n$ is the minimal number of edges that could be identified. 

Consider a GRN that there is an edge from each vertex in $\mathcal{S}_2=\{V_{n/2+1},\ldots,V_{n}\}$ to each vertex in $\mathcal{S}_1=\{V_1,\ldots,V_{n/2}\}$, and there is an edge from each vertex in $\mathcal{S}_1$ to $V_0$. All $n^2/4+n/2$ edges in this GRN can be identified. We guess $n^2/4+n/2$ is the maximal number of edges that could be identified. 

See Subsection \ref{81} and Appendix \ref{apppbp} for the performance of this method on simulated data and corresponding discussions. 

\begin{remark}
	Assume we have identified an edge $V_i\to V_j$. From the above argument, all edges in the shortest path from $V_j$ to $V_0$ can be identified. We add interventions on all genes except $V_i$ and genes on this shortest path. If the intervention on $V_i$ and the intervention on $V_j$ have the same sign on the change of $V_0$, then $V_i$ activates $V_j$. If they have different signs, then $V_i$ inhibits $V_j$. Here we assume that the effect of intervention is known to be decreasing gene expression.
\end{remark}

\subsection{Use ancestor-descendant relation to infer {GRN} structure with gene expression data}
\label{ADR}
Consider an unknown GRN with genes $V_1,\ldots,V_n$, and assume the GRN is a DAG. We can add interventions on any genes and measure their expression levels. If the intervention on $V_i$ leads to the change of $V_j$, then $V_i$ is an ancestor of $V_j$, and there is a directed path from $V_i$ to $V_j$ on the graph. With such ancestor-descendant relations, we can partially infer the GRN structure: at least $n-1$ edges can be identified. This method is rather elementary, but it is model-free.

For the two GRNs in Fig. \ref{exad}, $V_1$ has descendants $V_2,V_3,V_4$; $V_2$ has descendants $V_3,V_4$; $V_3$ has descendant $V_4$; $V_4$ has no descendant. Thus they are equivalent in the sense of ancestor-descendant relations. If a directed edge appears in all equivalent GRNs, we can determine that this edge exists; if a directed edge appears in none of equivalent GRNs, we can determine that this edge does not exist; if a directed edge appears in some but not all equivalent GRNs, we cannot determine whether this edge exists.

\begin{figure}
	\center
	$\xymatrix{
		V_1\ar[d]&V_3\ar[r]&V_4&&V_1\ar[r]\ar[d]&V_3\ar[r]&V_4\\
		V_2\ar[ru]&&&&V_2\ar[ru]&&
	}$\\
	\caption{Two equivalent GRNs in the sense of ancestor-descendant relations.}
	\label{exad}
\end{figure}
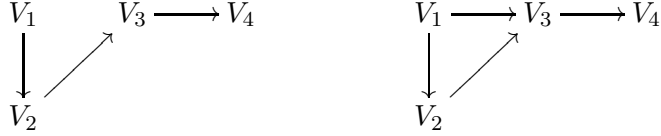

\begin{proposition}
	\label{adrp}
	The following procedure describes how to determine certain edges. (1) If $V_j$ is not a descendant of $V_i$, then we can determine that the edge $V_i\to V_j$ does not exist. (2) If $V_j$ is a descendant of $V_i$, and $V_i$ has another descendant $V_k$, which is an ancestor of $V_j$, then we cannot determine the existence of the edge $V_i\to V_j$. (3) If $V_j$ is a descendant of $V_i$, and $V_i$ does not have another descendant $V_k$, which is an ancestor of $V_j$, then we can determine that the edge $V_i\to V_j$ exists. 
\end{proposition}

\begin{proof}
	(1) If there is an edge $V_i\to V_j$, then $V_j$ is a descendant of $V_i$.
	
	(2) Since $V_i$ is an ancestor of $V_k$, we can add an edge $V_i\to V_k$ if it does not exist, and the ancestor-descendant relations are not affected. Since $V_k$ is an ancestor of $V_j$, we can add an edge $V_k\to V_j$ if it does not exist, and the ancestor-descendant relations are not affected. Since we have a path $V_i\to V_k \to V_j$, if there is an edge $V_i\to V_j$, we can delete it without affecting the ancestor-descendant relations. Since $V_i$ is an ancestor of $V_k$, if the edge $V_i\to V_k$ does not exist, we can add it, and the ancestor-descendant relations are not affected.
	
	(3) If there is no edge $V_i\to V_j$, then there is a path $V_i\to V_k\to \cdots\to V_j$. Thus $V_k$ is a descendant of $V_i$, and $V_j$ is a descendant of $V_k$, a contradiction.
\end{proof}

Consider $\mathcal{S}_1=\{V_1,\ldots,V_{n/2}\}$, $\mathcal{S}_2=\{V_{n/2+1},\ldots,V_n\}$. There is an edge from each vertex in $\mathcal{S}_1$ to each vertex in $\mathcal{S}_2$, and no other edge exists. Then all these $n^2/4$ edges can be identified. We guess this is the maximal number of edges that could be identified. When the GRN is a DAG, we can see that identified edges cannot form a triangle: if $V_i\to V_j$ and $V_j\to V_k$ can be identified, then $V_i\to V_k$ cannot be identified. Due to Tur\'an's Theorem \cite{G173}, a triangle-free graph with $n$ vertices can have at most $n^2/4$ edges. Therefore, in the DAG case, $n^2/4$ is the maximal number of edges that could be identified. 

If the GRN has cycles, we might recover no edge. For example, consider a GRN $V_1\to V_2\to V_3\to V_4 \to V_5$ with an edge $V_5\to V_1$, and a GRN $V_1\to V_3\to V_5\to V_2 \to V_4$ with an edge $V_4\to V_1$. In both GRNs, any vertex is a descendant of any vertex. Therefore, these two GRNs have the same ancestor-descendant relations, but share no edge.

\begin{proposition}
	If the GRN is a DAG with $n$ vertices, then we can use ancestor-descendant relations to identify at least $n-1$ edges. 
\end{proposition}
\begin{proof}
	The DAG is associated with a natural partial order: if there is an edge $V_i\to V_j$, then $V_i<V_j$. This partial order can be extended to a total order. If $V_i<V_j$ in the total order, then $V_i$ cannot be a descendant of $V_j$. 
	
	When $n=2$, this proposition is trivial. Assume this proposition is true for any $n<K$. When $n=K$, find a vertex $V_i$ which has no child. Assume that after deleting $V_i$ and edges linked to $V_i$, the DAG is divided into $m$ connected components $\mathcal{S}_1,\ldots,\mathcal{S}_m$. For the parents of $V_i$ in each $\mathcal{S}_j$, choose the largest one $V_j'$ under the total order. There is no other path from $V_j'$ to $V_i$. Thus the edge $V_j'\to V_i$ can be identified. With $m$ connected components, we can identify $m$ edges that lead to $V_i$. An edge that can be identified in each connected component should be able to be identified in the original DAG. Thus deleting $V_i$ does not affect identifying edges in each connected component. Apply this proposition to each connected component with size $n_j$, then we can identify at least $n_j-1$ edges. In sum, we can identify at least $\sum_j n_j-m+m=K-1$ edges.
\end{proof}

See Subsection \ref{82} and Appendix \ref{appadr} for the performance of this method on simulated data and corresponding discussions. 

\begin{remark}
	Assume we have identified an edge $V_i\to V_j$. This means no $V_k$ is a descendant of $V_i$ and an ancestor of $V_j$. Thus $V_i\to V_j$ is the only directed path from $V_i$ to $V_j$. When we intervene with $V_i$, so that the expression level of $V_i$ decreases, we can measure the expression level of $V_j$. If the expression level of $V_j$ decreases, then $V_i$ activates $V_j$. Otherwise, $V_i$ inhibits $V_j$.
\end{remark}

\subsection{Use conditional independence and ancestor-descendant relation to infer {GRN} structure}
\label{Sce2}
Assume the GRN is a DAG. We can measure the joint distribution of gene expression levels, and this distribution is Markov and faithful to the DAG. We can also intervene with genes to obtain the ancestor-descendant relations. With the conditional independence relations, we can find all edges of this DAG, but some edge directions cannot be determined. With the ancestor-descendant relations, we can determine the direction of each edge. Therefore, combining these two data types, we can reconstruct the full DAG.

See Subsection \ref{83} and Appendix \ref{appci} for the performance of this method on simulated data and corresponding discussions. 

\begin{remark}
	When the full DAG is inferred, we can use the same methods in Subsection \ref{DAG} to determine whether an edge $V_i\to V_j$ is activation or inhibition.
\end{remark}

\subsection{Use partially observed {ODE} to infer {GRN} structure with phenotype data}
\label{partial}
Consider an unknown GRN with phenotype $V_0$ and genes $V_1,\ldots,V_n$. Assume the GRN is a DAG. We can intervene with any genes, but we can only measure the phenotype on bulk level, not the genes. We assume there is a directed path from each $V_i$ to $V_0$, and there is no edge starting from $V_0$. Denote the levels of $V_0,V_1,\ldots,V_n$ by $x_0(t),x_1(t),\ldots,x_n(t)$, and assume they satisfy a linear ODE system. There is a partial order on $V_0,V_1,\ldots,V_n$, associated with edges in this DAG: if there is an edge $V_i\to V_j$, then stipulate that $V_i>V_j$. By the order-extension principle, we can extend it into a total order. We can reorder the genes by this total order, so that in the linear ODE system $\mathrm{d}\vec{x}/\mathrm{d}t=A\vec{x}+\vec{b}$, $a_{ij}>0$ means $i<j$, and $A$ is upper-triangular \cite{takeuti2013axiomatic}. This means there is no edge from $V_i$ to $V_j$ if $i<j$. The diagonal elements of $A$ are the degradation rates, which are negative. Therefore, the linear system is sign stable \cite{bone1988qualitative}, so that the unique fixed point $-A^{-1}\vec{b}$ is stable. 

Assume the system starts from the fixed point, and each time, we intervene with the value of one gene $V_i$. This means that $\vec{x}(0)=-A^{-1}\vec{b}+\vec{\delta}_i$, where $\vec{\delta}_i$ is a zero vector except its $i$th component. After the intervention, we can observe how $x_0(t)$ changes. Here we need to assume that the interventions added on different genes only change the initial values $x_i(0)$, but not system parameters $\{a_{ij}\}$ and $\{b_i\}$. We consider a question: based on only $x_0(t)$ under different interventions, can we partially infer the DAG? Surprisingly, the answer is yes (theoretically). We can determine whether a gene $V_i$ is an ancestor of another gene $V_j$.

Assume the diagonal elements of $A$ are different. Then the solution of $x_0$ is 
\[x_0(t)=c_{0}e^{a_{00}t}+c_{1}e^{a_{11}t}+\cdots+c_{n}e^{a_{nn}t}+d,\]
where the coefficients $c_i$ and $d$ depend on $A,\vec{b},\vec{x}(0)$. Assume in $V_0,V_1,\ldots,V_n$, the descendants of $V_i$ are $V_0,V_1,\ldots,V_{i-1}$. Since we start from the fixed point, except that $V_i$ is perturbed, those $V_j$ with $j>i$ shall be fixed, and the solution of $x_0$ is 
\[x_0(t)=c_{0}e^{a_{00}t}+c_{1}e^{a_{11}t}+\cdots+c_{i}e^{a_{ii}t}+d.\]
The coefficients $c_0,c_1,\ldots,c_i$ are all non-zero, unless the coefficients $a_{ij}$ satisfy an algebraic equation system. For example, consider the system 
\begin{equation*}
	\begin{split}
		\mathrm{d}x_0/\mathrm{d}t &=a_{00}x_0+a_{01}x_1+a_{02}x_2+b_0,\\
		\mathrm{d}x_1/\mathrm{d}t &=a_{11}x_1+a_{12}x_2+b_1,\\
		\mathrm{d}x_2/\mathrm{d}t &=a_{22}x_2+b_2.
	\end{split}
\end{equation*}
When we start from stationary and perturb $x_2$, the conditions for $x_0(t)=c_{0}e^{a_{00}t}+c_{1}e^{a_{11}t}+c_{2}e^{a_{22}t}+d$ with nonzero $c_0,c_1,c_2$ are (1) $a_{00}\ne a_{11}$; (2) $a_{00}\ne a_{22}$; (3) $a_{11}\ne a_{22}$; (4) $a_{01}a_{12}\ne 0$; (5) $a_{01}a_{12}+a_{02}a_{22}-a_{00}a_{02}\ne 0$; (6) $a_{01}a_{12}a_{22}-a_{01}a_{12}a_{11}-a_{02}a_{11}a_{22}+a_{02}a_{00}a_{11}+a_{02}a_{00}a_{22}-a_{02}a_{00}^2\ne 0$. In general, by Fubini-Tonelli theorem, matrix $A$ that does not satisfy such an algebraic equation system composes a zero-measured set in $\mathbb{R}^{(n+1)\times (n+1)}$ \cite{zhou2014multi}. Therefore, when we perturb $x_i(0)$, we can observe $x_0(t)$ as a linear combination of $i+1$ exponential functions. We can numerically determine the values of $a_{00},a_{11},\ldots,a_{ii}$. This means that when we perturb $x_i$, we can find $a_{00},a_{11},\ldots,a_{ii}$, which should correspond to $V_i$ and its descendants. The question is to find the correspondence between $a_{ii}$ and $V_i$. If a component $a_{nn}$ only appears after perturbing $x_n$, then we can make sure that $a_{n}$ corresponds to $V_n$. Then if $a_{jj}$ only appears after perturbing $x_n$ and $x_j$, then we know that $a_{jj}$ corresponds to $V_j$. Since there is a total order on $V_0,V_1,\ldots,V_n$, determined by the underlying DAG, we can determine the correspondence one by one. With the established correspondence, for each $V_i$, we can determine which genes are the descendants of $V_i$, by observing which components appear in $x_0(t)$. 

For example, consider one phenotype $V_0$ and three genes $V_1,V_2,V_3$, and we can measure $x_0(t)$ after perturbing any one gene. If we perturb $x_1$, $x_0(t)$ has two exponential components $\lambda_0,\lambda_1$. If we perturb $x_2$, $x_0(t)$ has two exponential components $\lambda_0,\lambda_2$. If we perturb $x_3$, $x_0(t)$ has four exponential components $\lambda_0,\lambda_1,\lambda_2,\lambda_3$. Since component $\lambda_3$ only appears after perturbing $x_3$, $\lambda_3$ corresponds to $V_3$. Since component $\lambda_2$ only appears after perturbing $x_2$ or $x_3$, and $V_3$ already corresponds to $\lambda_3$, we know that $\lambda_2$ corresponds to $V_2$. Since component $\lambda_1$ only appears after perturbing $x_1$ or $x_3$, and $V_3$ already corresponds to $\lambda_3$, we know that $\lambda_1$ corresponds to $V_1$. Finally, $\lambda_0$ corresponds to $V_0$. Therefore, we know that $V_1$ has descendant $V_0$; $V_2$ has descendant $V_0$; $V_3$ has descendants $V_0,V_1,V_2$. We can determine that the DAG has edges $V_3\to V_2$, $V_3\to V_1$, $V_2\to V_0$, $V_1\to V_0$, with a possible edge $V_3\to V_0$.

In sum, assume that we can observe the phenotype variable $x_0(t)$ under the intervention on each ancestor gene $V_i$ of $V_0$. Then we can determine whether one gene $V_i$ is an ancestor of another gene $V_j$. Using the same method in Subsection \ref{ADR}, this ancestor-descendant relation can partially determine the DAG. If the ODE system has nonlinear terms, the decomposition fails, and we cannot determine the perturbation on one gene is transmitted through other genes. Thus this method cannot work in nonlinear cases.

See Subsection \ref{84} and Appendix \ref{apppode} for the performance of this method on simulated data and corresponding discussions. This method is not numerically feasible, and not applicable in reality. 

\begin{remark}
	We can only observe the phenotype, not the genes, although we can intervene with genes. Consider a GRN $V_2\to V_1\to V_0$ with an edge $V_2\to V_0$. From the ancestor-descendant relations, we can only identify the edges $V_2\to V_1$ and $V_1\to V_0$. Assume $V_1$ activates $V_0$. If the level of $V_0$ decreases as the level of $V_2$ decreases (by an intervention), then we can explain that $V_2$ strongly activates $V_0$, and $V_2$ weakly activates or inhibits $V_1$. If the level of $V_0$ increases as the level of $V_2$ decreases (by an intervention), then we can explain that $V_2$ strongly inhibits $V_0$, and $V_2$ weakly activates or inhibits $V_1$. Therefore, we cannot determine whether $V_2\to V_1$ is activation or inhibition.
\end{remark}

\subsection{Use path blocking relation and partially observed {ODE} to infer {GRN} structure with phenotype data}
\label{LPD}
Consider an unknown GRN with genes $V_1,\ldots,V_n$ and a phenotype $V_0$. Assume the GRN is a DAG. We can intervene with any genes, but we can only measure the phenotype on bulk level, not the genes. We assume there is a directed path from each $V_i$ to $V_0$, and there is no edge starting from $V_0$. Assume the path blocking property holds, and the levels of $V_0,V_1,\ldots,V_n$ satisfy a linear ODE system. Using the method in Subsection \ref{PBP}, we can infer some edges. Using the method in Subsection \ref{partial}, we can also infer some edges. However, combining these two results might not fully recover the full GRN. For example, the two GRNs in Fig. \ref{ex16} have the same ancestor-descendant relations and path blocking relations. Thus we cannot determine whether $V_1\to V_3$ exists in this setting.

See Subsection \ref{85} and Appendix \ref{apppba} for the performance of this method on simulated data and corresponding discussions.

\begin{figure}
	\center
	$\xymatrix{
		V_1\ar[d]\ar[r]\ar[dr]&V_4\ar[r]&V_0&&V_1\ar[d]\ar[r]&V_4\ar[r]&V_0\\
		V_2\ar[r]&V_3\ar[u]&&&V_2\ar[r]&V_3\ar[u]&
	}$\\
	\caption{Two equivalent GRNs in the sense of ancestor-descendant relations and path blocking relations.}
	\label{ex16}
\end{figure}
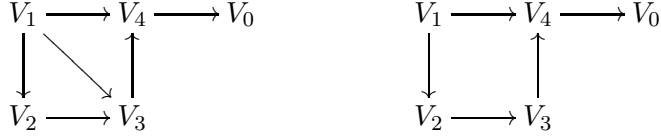

\begin{remark}
	For the two GRNs in Fig. \ref{ex16}, we can identify edges $V_1\to V_2$ and $V_2\to V_3$, but not the edge $V_1\to V_3$. Since $V_2\to V_3\to V_4\to V_0$ is the only directed path from $V_2$ to $V_0$, if we need to observe the influence of $V_2$ on $V_0$, or the influence of $V_1$ on $V_2$, we cannot intervene with $V_3$ or $V_4$ (otherwise, the path is blocked). Therefore, to determine the effect of $V_1$ on $V_2$, we can only interfere with $V_1$ or $V_2$. Similar to the discussions in Subsection \ref{partial}, we cannot determine $V_1\to V_2$ is activation or inhibition.
\end{remark}

\section{Implementations of novel methods}
\label{imp}
In this section, we implement our novel methods in Section \ref{novel}, and evaluate them on simulated data. Further discussions and details are in Appendix \ref{app}. The motivation of inventing novel methods is to fill the blanks of Table \ref{Tab}, namely those scenarios that are not thoroughly-studied. It is difficult to obtain data for such scenarios. Thus we only evaluate our methods on synthetic data. 

\subsection{Implementation of the path blocking method with genes}
\label{80}
In this subsection, we test the method introduced in Subsection \ref{PBG} on simulated data. To determine whether an edge $V_1\to V_2$ exists, we intervene with all other genes, and test whether perturbing $V_1$ has extra effect on the expression of $V_2$. Since all other genes are perturbed and useless, we can ignore them and focus on $V_1$ and $V_2$.

Since the system is simple enough, we propose an explicit mechanism for gene regulation and noise, and study the effect of noise on GRN inference. In each simulation, randomly generate the regulation coefficients for $V_1$ and $V_2$, and solve the level of $V_2$, $x_2$. Then intervene with $V_1$ and solve $\bar{x}_2$. We add independent noises on $x_2$ and $\bar{x}_2$ to obtain $x_2'$ and $\bar{x}_2'$, where the magnitude does no exceed $N$. We check whether the ratio $x_2'/\bar{x}_2'$ exceeds the range $[1-T,1/(1-T)]$, where $T$ is the threshold. If yes, then the edge $V_1\to V_2$ exists; otherwise, $V_1$ to $V_2$ is blocked by all other genes, and the edge $V_1\to V_2$ does not exist.

For different values of noise level $N$ and threshold $T$, we simulate $10^6$ times and calculate the percentage that the existence of edge $V_1\to V_2$ can be correctly inferred. See Fig. \ref{f81} for the contour map of the correct rate. The red curve in Fig. \ref{f81} indicates the optimal choice of threshold $T$ for each noise level $N$. See Appendix \ref{apppbg} for simulation details.

\begin{figure}[]
	\centering
	\includegraphics[width=5.4in]{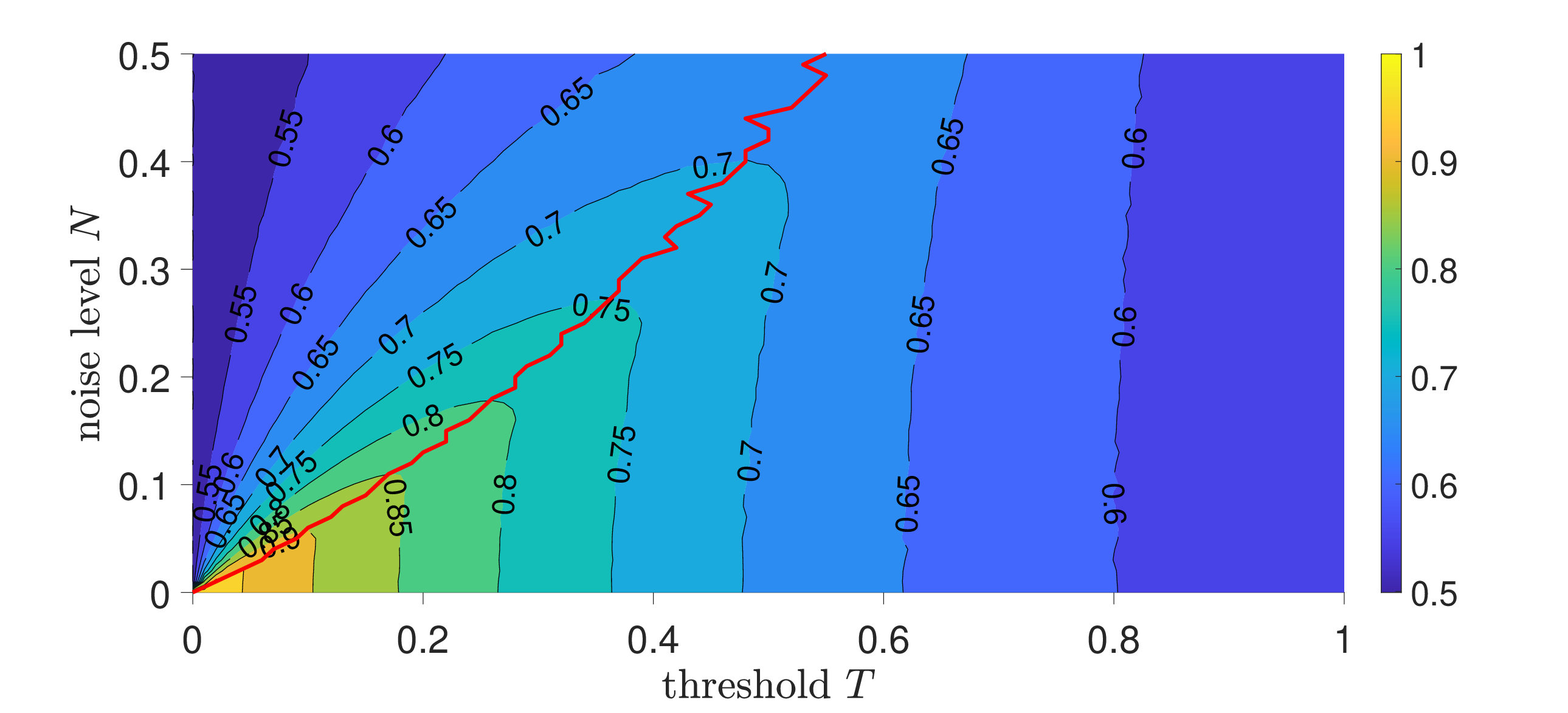}
	\caption{The contour plot for the correct rate of the path blocking method with genes (Subsection \ref{80}), with respect to different threshold $T$ ($x$-axis) and noise level $N$ ($y$-axis). The color represents the correct rate. The red curve indicates the optimal threshold for different values of noise level $N$.}
	\label{f81}
\end{figure}

\subsection{Implementation of the path blocking method with a phenotype}
\label{81}
In this subsection, we test the method introduced in Subsection \ref{PBP} on simulated data. Since different mechanisms for gene regulation and noise lead to different error rates in determining path blocking relations, we only study how the error in path blocking relations affect the inferred GRN.

In each simulation, randomly generate a GRN with three genes $V_1,V_2,V_3$ and one phenotype $V_0$. There is no edge from $V_0$ to $V_i$. Using the GRN, construct the path blocking relations from each gene $V_i$ to the phenotype $V_0$. Using Proposition \ref{pbpt}, we can partially reconstruct the GRN $\mathcal{G}$, where each edge $V_i\to V_j$ has three possibilities: $V_i\to V_j$ exists, $V_i\to V_j$ does not exist, the existence of $V_i\to V_j$ cannot be determined. 

In the above simulation, we add noise to each true path blocking relation: if $\mathcal{S}$ does not block $V_i$ to $V_0$, with probability $p$, we observe that $\mathcal{S}$ blocks $V_i$ to $V_0$; if $\mathcal{S}$ blocks $V_i$ to $V_0$, with probability $q$, we observe that $\mathcal{S}$ does not block $V_i$ to $V_0$. With the observed path blocking relations, we use Proposition \ref{pbpt} again to reconstruct the GRN $\mathcal{G}'$. 

Define $r_0$ to be the number of edges that do not exist in $\mathcal{G}$; define $r_1$ to be the number of edges that exist in $\mathcal{G}$; define $s_0$ to be the number of edges that do not exist in $\mathcal{G}'$; define $s_1$ to be the number of edges that exist in $\mathcal{G}'$; define $t_0$ to be the number of edges that exist in neither $\mathcal{G}$ nor $\mathcal{G}'$; define $t_1$ to be the number of edges that exist in both $\mathcal{G}$ and $\mathcal{G}'$. Edges that cannot be determined are ignored. Then we can define four performance measures: sensitivity $\text{SEN}=t_1/r_1$, specificity $\text{SPE}=t_0/r_0$, positive predictive value $\text{PPV}=t_1/s_1$, negative predictive value $\text{NPV}=t_0/s_0$. In the above example, $r_0=3,r_1=3,s_0=3,s_1=3,t_0=2,t_1=1$. Thus $\text{SEN}=1/3,\text{SPE}=2/3,\text{PPV}=1/3,\text{NPV}=2/3$.

The above four performance measures are functions of the error rates $p,q$. See Fig. \ref{f82} for the contour maps of these four measures, averaged over $10^5$ simulations. See Appendix \ref{apppbp} for simulation details.

\begin{figure}[]
	\centering
	\includegraphics[width=5.4in]{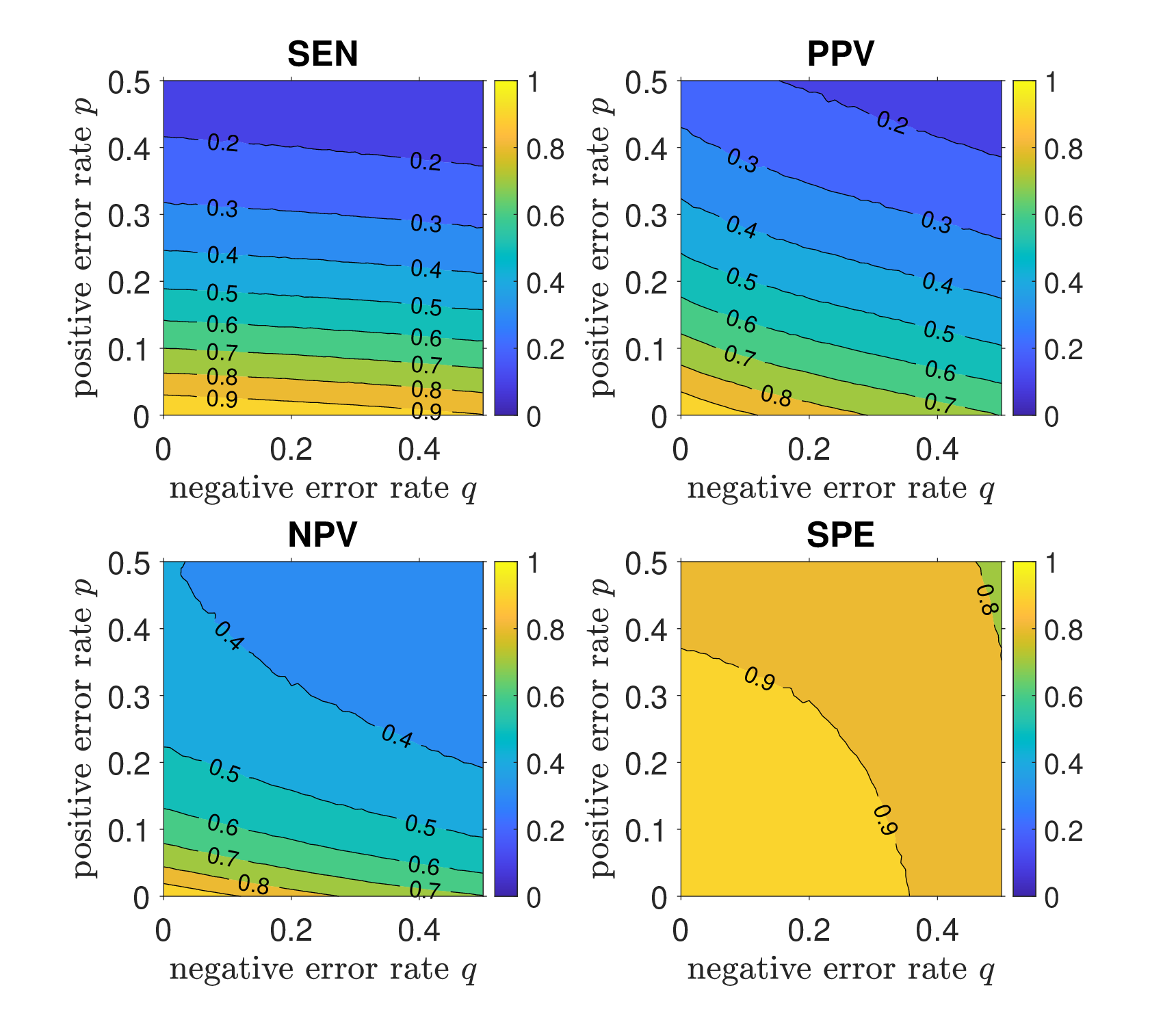}
	\caption{The contour plots for the four performance measures of the path blocking method with a phenotype (Subsection \ref{81}), with respect to different negative error rate $q$ ($x$-axis) and positive error rate $p$ ($y$-axis). SEN is sensitivity; PPV is positive predictive value; NPV is negative predictive value; SPE is specificity. The color represents the corresponding rate.}
	\label{f82}
\end{figure}

\subsection{Implementation of the ancestor-descendant relation method}
\label{82}
In this subsection, we test the method introduced in Subsection \ref{ADR} on simulated data. Similar to Subsection \ref{81}, we focus on the effect of errors in ancestor-descendant relations to the inference of GRN, not concrete regulation mechanisms.

In each simulation, randomly generate a GRN with 10 genes. The generating method guarantees that the GRN is a DAG with 10-30 edges. Using the DAG, we construct a $10\times 10$ matrix $A$ that describes the ancestor-descendant relations. If $V_i$ is an ancestor of $V_j$, then $A_{ij}=1$; otherwise, $A_{ij}=0$. The diagonal elements of $A$ are stipulated to be $-1$. By Proposition \ref{adrp}, we can use $A$ to partially infer the connectivity matrix $C$: if we can confirm the edge $V_i\to V_j$ exists, set $C_{ij}=1$; if we can confirm the edge $V_i\to V_j$ does not exist, set $C_{ij}=0$; if we cannot determine the existence of the edge $V_i\to V_j$, set $C_{ij}=0.5$. The diagonal elements of $C$ are stipulated to be $-1$. 

In the above simulation, we assume the observed ancestor-descendant relation matrix $A'$ is the true $A$ with some perturbations: if $A_{ij}=1$, with probability $p$, set $A'_{ij}=0$; if $A_{ij}=0$, with probability $q$, set $A'_{ij}=1$; otherwise, set $A'_{ij}=A_{ij}$. With the perturbed $A'$, we use Proposition \ref{adrp} to construct $C'$. 

Similar to Subsection \ref{81}, we define the four performance measures by comparing $C$ and $C'$: sensitivity SEN, specificity SPE, positive predictive value PPV, negative predictive value NPV. These four performance measures are functions of the error rates $p,q$. See Fig. \ref{f83} for the contour maps of these four measures, averaged over $10^4$ simulations. See Appendix \ref{appadr} for simulation details.

\begin{figure}[]
	\centering
	\includegraphics[width=5.4in]{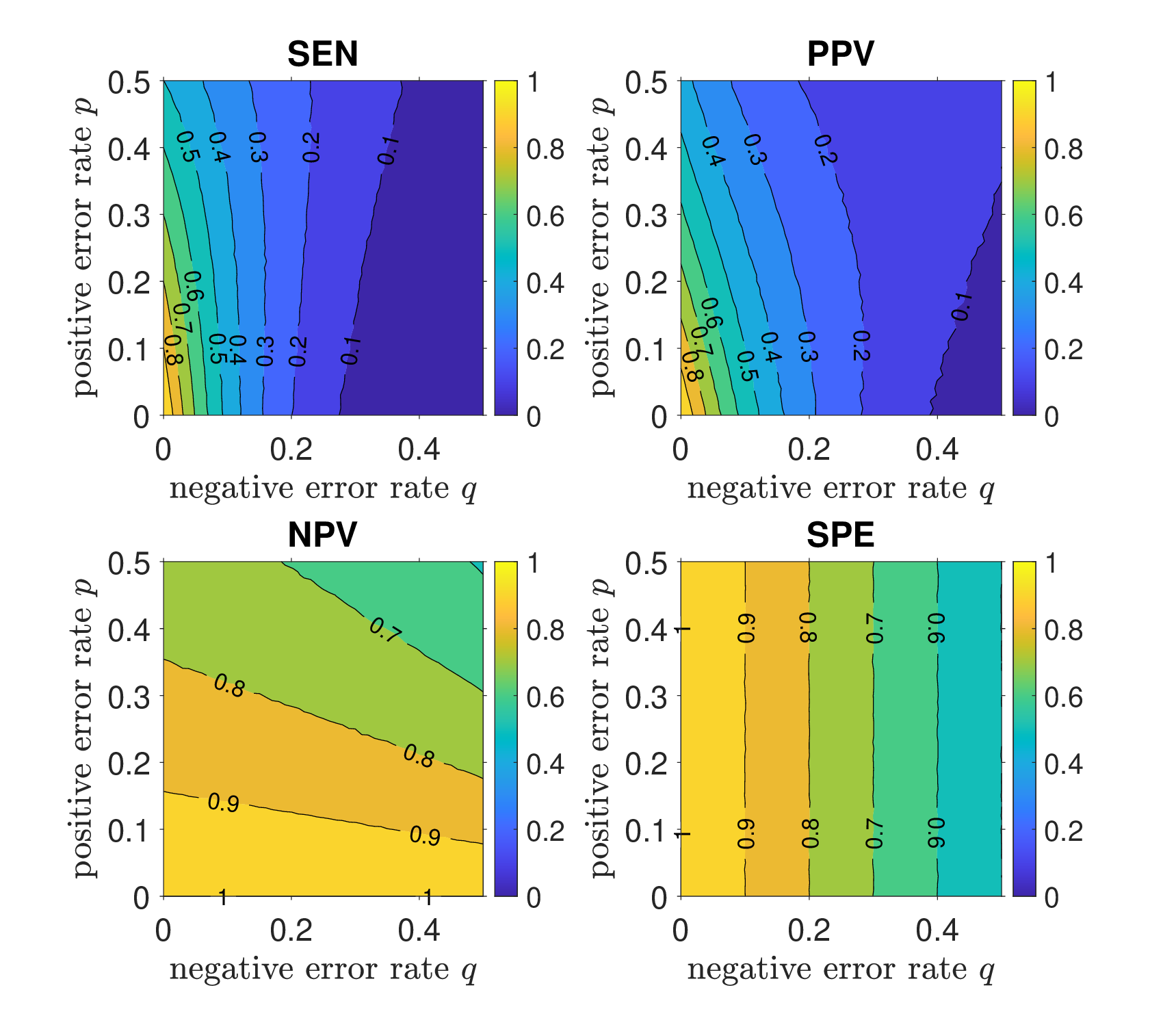}
	\caption{The contour plots for the four performance measures of the ancestor-descendant relation method (Subsection \ref{82}), with respect to different negative error rate $q$ ($x$-axis) and positive error rate $p$ ($y$-axis). SEN is sensitivity; PPV is positive predictive value; NPV is negative predictive value; SPE is specificity. The color represents the corresponding rate.}
	\label{f83}
\end{figure}

\subsection{Implementation of the conditional independence + ancestor-descendant method}
\label{83}
In this subsection, we test the method introduced in Subsection \ref{Sce2} on simulated data. Similar to Subsection \ref{81}, we focus on the effect of errors in conditional independence and ancestor-descendant relations to the inference of GRN, not concrete regulation mechanisms.

In each simulation, randomly generate a GRN with 3 genes, and make sure that it is a connected DAG. For the generated DAG, determine all the (conditional) independence relations and ancestor-descendant relations. 

Since we have interventional single-cell gene expression data, ``$V_i$ is an ancestor of $V_j$'' is equivalent to ``after perturbing $V_i$, $X_i$ and $X_j$ are dependent''. Thus both relations are essentially determined by independence tests, and we can assume that they have the same (positive or negative) error rate. In the above simulation, we assume each pair of (conditionally) independent variables has probability $q$ to be dependent in the observation, and each pair of (conditionally) dependent variables has probability $p$ to be independent in the observation. With the perturbed relations, we infer the DAG and compare with the true DAG. 

Similar to Subsection \ref{81}, we define the four performance measures: sensitivity SEN, specificity SPE, positive predictive value PPV, negative predictive value NPV. These four performance measures are functions of the error rates $p,q$. See Fig. \ref{f84} for the contour maps of these four measures, averaged over $10^5$ simulations. See Appendix \ref{appci} for simulation details.

\begin{figure}[]
	\centering
	\includegraphics[width=5.4in]{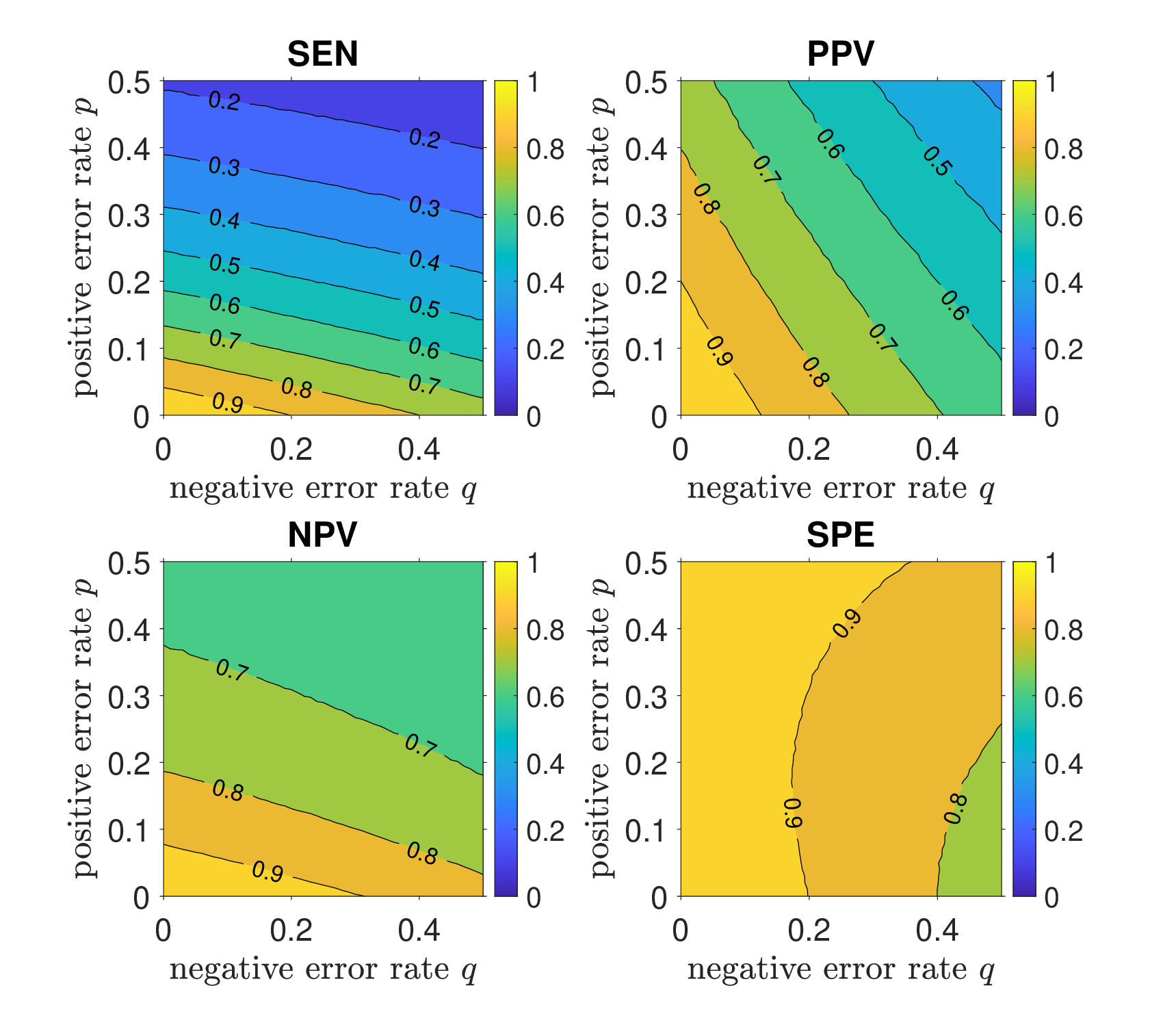}
	\caption{The contour plots for the four performance measures of the conditional independence + ancestor-descendant method (Subsection \ref{83}), with respect to different negative error rate $q$ ($x$-axis) and positive error rate $p$ ($y$-axis). SEN is sensitivity; PPV is positive predictive value; NPV is negative predictive value; SPE is specificity. The color represents the corresponding rate.}
	\label{f84}
\end{figure}

\subsection{Implementation of the partially observed ODE method}
\label{84}
In this subsection, we test the method introduced in Subsection \ref{partial} on simulated data. 

We start with a GRN with two genes $V_1,V_2$ and a phenotype $V_0$, whose levels $x_1(t)$, $x_2(t)$, $x_0(t)$ satisfy a linear ODE system. We start from stationary, and perturb $x_1$ or $x_2$, and observe $x_0(t)$. Since $x_0(t)=c_0e^{a_{00}t}+c_1e^{a_{11}t}+c_2e^{a_{22}t}+d$, we try to solve $a_{00},a_{11},a_{22}$ from the observed $x_0(t)$. When the observed $x_0(t)$ is accurate, we can barely solve $a_{00},a_{11},a_{22}$ numerically. When $x_0(t)$ is perturbed by a small noise (less than 1\%), $a_{00},a_{11},a_{22}$ already cannot be solved numerically. The reason is that $a_{00},a_{11},a_{22}$ satisfy a highly nonlinear equation system, and this system is sensitive to perturbations. Another approach is to apply the Laplace transform, which turns $c_0e^{a_{00}t}$ into $c_0/(s-a_{00})$. Although exponential functions are transformed into rational functions, it is still difficult to solve $a_{00},a_{11},a_{22}$ under a small noise.

Without $a_{00},a_{11},a_{22}$, we cannot determine the ancestor-descendant relations, and the GRN cannot be inferred. Therefore, this method is not applicable in reality, since it is numerically infeasible, and requires unreasonable accuracy for $x_0(t)$. Machine learning-based methods \cite{huang2020learning} might be able to save this idea. See Appendix \ref{apppode} for simulation details and further discussions.

\subsection{Implementation of the path blocking + ancestor-descendant method}
\label{85}
In this subsection, we test the method introduced in Subsection \ref{LPD} on simulated data. As discussed in Subsection \ref{84}, the partially observed ODE method fails to produce ancestor-descendant relations in reality. Nevertheless, we assume that we can determine the ancestor-descendant relations in another way (possibly with errors). Similar to Subsection \ref{81}, we focus on the effect of errors in path blocking and ancestor-descendant relations to the inference of GRN, not concrete regulation mechanisms.

In each simulation, randomly generate a DAG with three genes $V_1,V_2,V_3$ and one phenotype $V_0$. There is no edge from $V_0$ to $V_i$. Using the DAG, construct the path blocking relations from each gene $V_i$ to the phenotype $V_0$, and the ancestor-descendant relations for $V_i,V_j$. Use the path blocking relations (Proposition \ref{pbpt}) to infer the DAG $\mathcal{G}$. If some edges cannot be determined, use the ancestor-descendant relations (Proposition \ref{adrp}) to check if they can be determined. 

Notice that ``$V_i$ is an ancestor of $V_j$'' and ``$V_i$ to $V_0$ cannot be blocked'' both mean that ``the change of $V_i$ affects $V_j$''. Thus both relations are essentially determined by testing whether the change is significant, and we can assume that they have the same (positive or negative) error rate. In the above simulation, we assume each pair of ``no change'' variables (blocking or is not an ancestor) has probability $q$ to be inverted in the observation, and each pair of ``change'' variables (not blocking or is an ancestor) has probability $p$ to be inverted in the observation. With the perturbed relations, we infer the DAG $\mathcal{G}'$ and compare with $\mathcal{G}$. 

Similar to Subsection \ref{81}, we define the four performance measures: sensitivity SEN, specificity SPE, positive predictive value PPV, negative predictive value NPV. These four performance measures are functions of the error rates $p,q$. See Fig. \ref{f85} for the contour maps of these four measures, averaged over $10^5$ simulations. See Appendix \ref{apppba} for simulation details.

\begin{figure}[]
	\centering
	\includegraphics[width=5.4in]{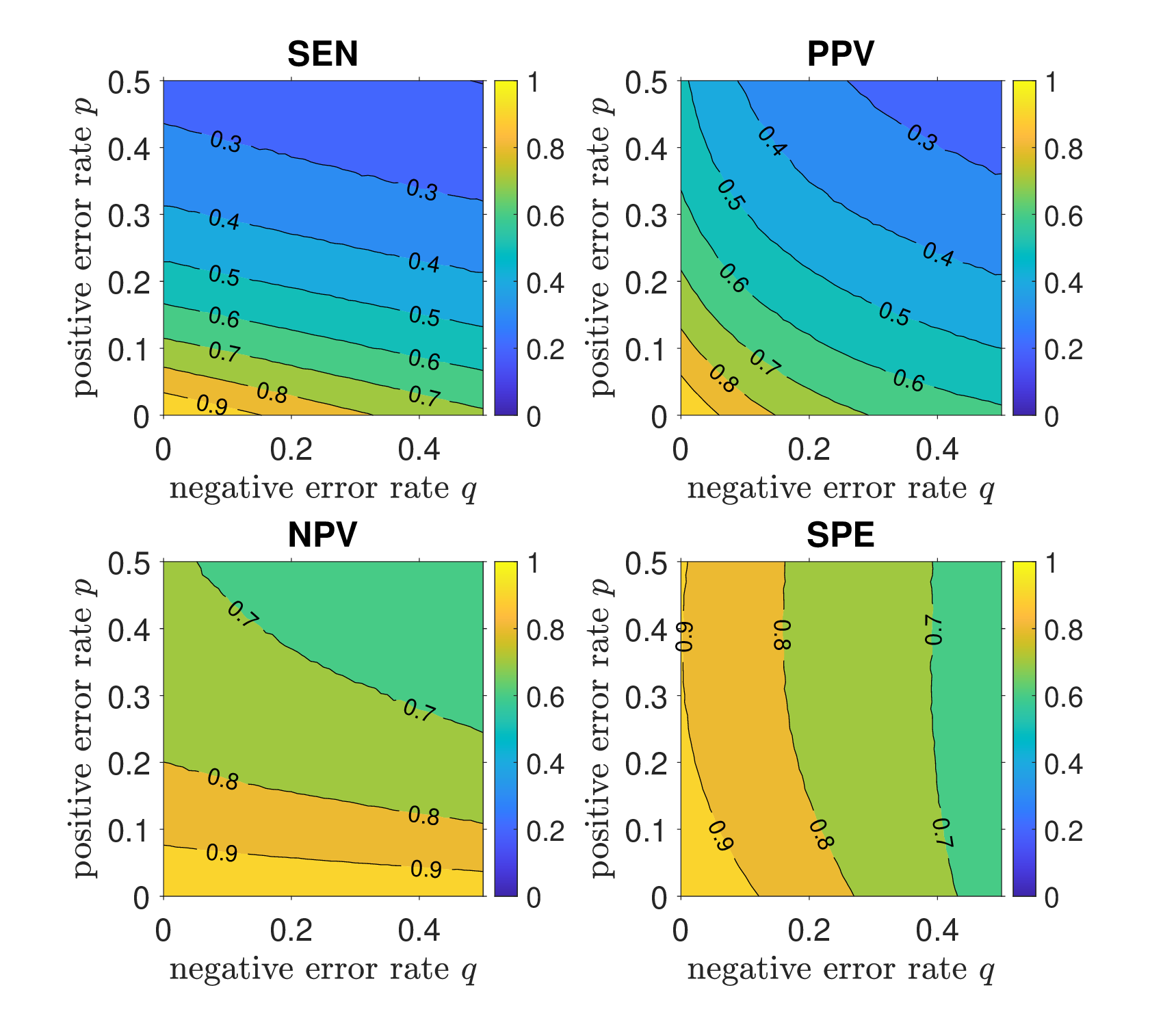}
	\caption{The contour plots for the four performance measures of the path blocking + ancestor-descendant method (Subsection \ref{85}), with respect to different negative error rate $q$ ($x$-axis) and positive error rate $p$ ($y$-axis). SEN is sensitivity; PPV is positive predictive value; NPV is negative predictive value; SPE is specificity. The color represents the corresponding rate.}
	\label{f85}
\end{figure}

\section{Different scenarios of inferring {GRN} structure}
\label{scenarios}
In this section, we discuss the question that with each possible data type, whether the GRN structure can be inferred. As introduced in Section \ref{biol}, each scenario (data type) is described in the following four dimensions: (1) Do we measure \textbf{Gene} expression or \textbf{Phenotype}? (2) Do we measure on \textbf{Single-cell} level or \textbf{Bulk} level? (3) Do we measure at a single time point (\textbf{One-time}) or at multiple time points (\textbf{Time series})? (4) Is the measurement \textbf{Interventional} or \textbf{Non-interventional}? When the measurement is on the single-cell level and at multiple time points, we need to discuss whether we can measure the \textbf{Joint} distribution or \textbf{Marginal} distributions for multiple time points.

For each scenario, there are different approaches to infer the GRN structure, corresponding to different models and different assumptions. We do not aim at exhausting all existing and potential approaches.

\subsection{Scenario 1: gene, single-cell, one-time, non-interventional}
As discussed in Subsection \ref{DAG}, if the GRN is a DAG, and the joint distribution is Markov and faithful to this DAG, then we can partially determine the GRN structure: all edges can be identified, except that the directions of some edges are unknown. This method has been applied in determining GRN structures \cite{emmert2012statistical}. Various algorithms utilize experimental data to determine conditional independence relations \cite{Huang,Kun}, and then infer DAG structures \cite{Ju}. There are also other approaches to model this scenario. In a network, we can use neighboring relations to predict the characters associated with each node \cite{wang2021inference}. Reversely, we can use the node characters (gene expression profiles) to infer the network structures. This approach can only provide a rough guess among different possibilities.

\subsection{Scenario 2: gene, single-cell, one-time, interventional}
Assume the GRN is a DAG, and the stationary joint distribution is Markov and faithful to this DAG. From the interventions, we can obtain the ancestor-descendant relations. As discussed in Subsection \ref{Sce2}, we can combine these two kinds of relations to fully reconstruct the DAG.  If the DAG condition holds, as discussed in Subsection \ref{ADR}, we can use ancestor-descendant relations to partially determine the GRN. If the path blocking property holds under the interventions, similar to Scenario 6, we can fully determine the GRN.

\subsection{Scenario 3: gene, single-cell, time series, non-interventional}

\textbf{Scenario 3a: Joint}. If we can measure the joint distribution for different time points, then as discussed in Subsection \ref{GSP}, we can fully reconstruct the GRN, even if there are cycles. This method has been applied in determining GRN structures \cite{yu2004advances}. There are some other approaches, such as Granger causality \cite{BBS09} and Boolean networks \cite{shmulevich2002gene}. Readers may refer to a review paper \cite{penfold2011infer} for different methods applicable for this scenario.

\textbf{Scenario 3b: Marginal}. If we can only measure the marginal distribution for each time point, then this is essentially the same as Scenario 1. If the GRN is a DAG, and the joint distribution is Markov and faithful to this DAG, then we can partially determine the GRN structure.

\subsection{Scenario 4: gene, single-cell, time series, interventional}
\textbf{Scenario 4a: Joint}. If we can measure the joint distribution for different time points, then similar to Scenario 3, we can fully reconstruct the GRN.

\textbf{Scenario 4b: Marginal}. If we can only measure the marginal distribution for each time point, there are different approaches. If we take expectations for expression profiles, this scenario moves to Scenario 8, and we can use the ODE model to fully recover the GRN under the linearity assumption. If the path blocking property holds under the interventions, similar to Scenario 6, we can use the blocking relations to fully recover the GRN. If the GRN is a DAG, and the stationary joint distribution is Markov and faithful to this DAG, similar to Scenario 2, we can use the conditional independence relations and the ancestor-descendant relations to fully recover the GRN. If the GRN is a DAG, similar to Scenario 6, we can use the ancestor-descendant relations to partially recover the GRN. It is also possible to directly fit the marginal distributions to a stochastic model \cite{gao2022data}.

\subsection{Scenario 5: gene, bulk, one-time, non-interventional}
In this scenario, we can only observe the bulk level expression profile. Since there is no intervention, and the expression level of each gene is just a deterministic number, we cannot infer the structure of GRN. 

\subsection{Scenario 6: gene, bulk, one-time, interventional}

If the path blocking property holds under the interventions, as discussed in Subsection \ref{PBG}, there is an edge $V_i\to V_j$ if and only if they cannot be blocked by all other vertices. This means that we can fully determine the GRN even if it has cycles. If the GRN is a DAG, we can obtain the ancestor-descendant relation, and use the method in Subsection \ref{ADR} to partially infer the GRN structure.

\subsection{Scenario 7: gene, bulk, time series, non-interventional}
Without intervention, the cell population is in equilibrium. Thus multiple observations at different time points should provide the same information. Scenario 7 is essentially the same as Scenario 5. We cannot infer the GRN structure.

\subsection{Scenario 8: gene, bulk, time series, interventional}
Assume the expression levels satisfy a linear ODE system. As discussed in Subsection \ref{ODE}, we can use the observed data to calculate the parameters, then determine the GRN structure. This method is extensively used \cite{bansal2006inference,wang2013integration}. If the ODE system has known nonlinear terms, one can solve the parameters similarly, although the data set needs to have a larger size. We can also treat this like Scenario 6. If the path blocking property holds, we can fully determine the GRN. If the DAG assumption holds, we can partially infer the GRN structure through ancestor-descendant relations.

\subsection{Scenario 9: phenotype, single-cell, one-time, non-interventional}
In this scenario, we can neither perturb a gene nor measure the expression level of a gene. Without information about genes in the data, we cannot infer the structure of GRN.

\subsection{Scenario 10: phenotype, single-cell, one-time, interventional}
We assume the path blocking property holds from any gene to the phenotype: the intervention on one gene $V_i$ does not change the phenotype $V_0$, if and only if all paths from $V_i$ to $V_0$ are blocked by intervened genes. With the path blocking relation, as discussed in Subsection \ref{PBP}, we can use the path blocking relation to partially determine the GRN.

\subsection{Scenario 11: phenotype, single-cell, time series, non-interventional}
\textbf{Scenario 11a/b: Joint/Marginal}. Similar to Scenario 9, there is no information about genes in the data, and we cannot infer the structure of GRN. In this scenario, whether we can measure the joint distribution for multiple time points does not change the conclusion.

\subsection{Scenario 12: phenotype, single-cell, time series, interventional}
\textbf{Scenario 12a/b: Joint/Marginal}. This scenario is almost the same as Scenario 16. The only difference is that we measure single-cell level phenotype data, which have higher accuracy in detecting differences. We can partially infer the GRN if (1) the path blocking property holds; or (2) the dynamics is linear, and the GRN is a DAG; or (3) the path blocking property holds, the dynamics is linear, and the GRN is a DAG. Notice that in (2) and (3), for some identified edges, we cannot determine whether the regulatory relation is activation or inhibition. In this scenario, whether we can measure the joint distribution for multiple time points does not change the conclusion.

\subsection{Scenario 13: phenotype, bulk, one-time, non-interventional}
Similar to Scenario 9, there is no information about genes in the data, and we cannot infer the structure of GRN.

\subsection{Scenario 14: phenotype, bulk, one-time, interventional}
This scenario is almost the same as Scenario 10. The only difference is that we measure bulk level phenotype data, which have lower accuracy in detecting differences. We can partially determine the GRN.

\subsection{Scenario 15: phenotype, bulk, time series, non-interventional}
Similar to Scenario 9, there is no information about genes in the data, and we cannot infer the structure of GRN.

\subsection{Scenario 16: phenotype, bulk, time series, interventional}
If the path blocking property holds, then as discussed in Subsection \ref{PBP}, we can partially determine the GRN structure. If the dynamics is linear, and the GRN is a DAG, then as discussed in Subsection \ref{partial}, we can determine the ancestor-descendant relations, and partially determine the GRN structure. Notice that for some identified edges, we cannot determine whether the regulatory relation is activation or inhibition. If the path blocking property holds, the dynamics is linear, and the GRN is a DAG, then as discussed in Subsection \ref{LPD}, we can combine the above two methods and partially infer the GRN structure. Notice that for some identified edges, we cannot determine whether the regulatory relation is activation or inhibition.

\section{Discussion and conclusions}
\label{disc}
In this paper, we introduce the problem of GRN structure inference with experimental data. Depending on the data type, the whole problem is classified into 20 scenarios. For each scenario, we apply some assumptions and transform it into a well-defined mathematical problem. Then we either introduce known inference methods or propose new inference methods and test them on simulated data. Nevertheless, all the assumptions involved oversimplify the complicated gene regulation process in reality, so that the inferred results are not fully reliable.

This paper does not cover all mathematical approaches to the GRN structure inference problem. For instance, we do not cover a recent trend for applying deep learning in GRN structure inference \cite{yang2019predicting,zrimec2020deep,turki2021discriminating,shrivastava2022grnular,zheng2022accurate}. As machine learning and other related fields flourish, new mathematical methods for scenarios discussed in this paper will appear. With the development of biological technologies, there will be new data types that are not included in our 20 scenarios. This work just provides a general paradigm of inferring GRN (or general networks) structures: given what type of data, under what assumptions, what we can infer about the GRN structure. 

We mainly focus on theoretical inference methods, not on practical algorithms that work on real data. There is a gap between theory and algorithm, since different implementation approaches might introduce different levels of errors and have different efficiencies \cite{wang2020causal}. The designing of corresponding algorithms is related to some mathematical fields not covered in this paper: statistical inference \cite{casella2021statistical}, numerical computing \cite{pozrikidis1998numerical}, etc.

To derive ODE models, we assume there are infinitely many identical cells living in a stationary environment, and the underlying dynamics is time-homogeneous. These assumptions are not always true: we only have finitely many cells; cells have heterogeneity; the environment keeps changing; cells are mutating over time. Besides, this derivation might fail in nonlinear cases. For a nonlinear function $f$ and a random variable $X$, the expectation generally cannot cross $f$: $\mathbb{E}[f(X)]\ne f(\mathbb{E}X)$. Therefore, the ODE-based models cannot match the reality with extremely high accuracy, although some ODE-based models (such as the method in Subsection \ref{partial}) are numerically sensitive to perturbations.

When there are not enough data, one solution is to produce more data through interpolation \cite{bansal2006inference}. With this method, the dynamics of gene regulation is already stipulated by the data interpolation method. The new data can only be used to confirm that their dynamics follow the interpolation method. Therefore, this approach is essentially equivalent to adding more assumptions (e.g., linearity) into the model. We need to be cautious not to regard assumptions as conclusions.

In Scenario 5 and Scenario 7, we have the bulk level gene expression profiles without interventions. A common data type is to repeat the experiment on different populations of the same cell type, such as a portion of the DREAM challenge data \cite{prill2010towards} and the M3D database \cite{faith2007many}. Even without intervention, the expression levels of the same gene in different experiments (different populations) might differ. One might propose an approach that we can regard different expression levels of the same gene as samples from a random variable. Then the method in Section \ref{DAG} can be adopted. However, on bulk level, stochasticity has been averaged out, and the fluctuations in different experiments must come from different values of some unobserved variables (e.g., phenotype, environment). We provide a simple example to show that this approach might build a false relation between two independent genes. One hidden variable might affect multiple genes. Assume there is no edge between $V_1,V_2$, and they are both regulated by a hidden variable $H$ (confounder). We can observe that the expression levels of $V_1,V_2$ are always dependent, conditioned on any observable gene expression levels. Therefore, we obtain the false conclusion that there is an edge between $V_1$ and $V_2$. 

In Scenario 6 and Scenario 8, we have the bulk level gene expression profiles after different interventions. One approach of determining possible gene regulations is to directly calculate the correlation coefficient or mutual information of the expression levels of two genes \cite{hurley2012gene}. Here data after different interventions are regarded as samples from the same distribution. Assume there are three independent genes $V_1,V_2,V_3$. After intervention on $V_1$, the expression levels of $V_1,V_2$ are $0.5$ and $1$; after intervention on $V_2$, the expression levels of $V_1,V_2$ are $1$ and $0.5$; after intervention on $V_3$, the expression levels of $V_1,V_2$ are $1$ and $1$. Then the correlation coefficient of $V_1,V_2$ is $-0.5$, although $V_1,V_2$ are independent.

When the number of genes in the GRN increases, for various methods in this paper, required amount of experimental data and computation time also increase (might be exponentially fast). To avoid such problems, we need to assume the GRN is sparse, or just consider a few genes. Nevertheless, such solutions might lead to unreliable results, since GRNs in reality might be large and dense \cite{bansal2006inference}.

\section*{Acknowledgments}
Y.W. would like to thank Dr. Nadya Morozova and Dr. Andrey Minarsky for helpful discussions. This research was partially supported by NIH grant R01HL146552.

\appendix

\section{Discussions and details of simulation}
\label{app}
In this appendix, we describe the simulation details in Section \ref{imp}, and discuss the simulation results.

\subsection{Path blocking method with genes}
\label{apppbg}

In the equation system $\begin{bmatrix}
	-1 &    a_{12} \\
	a_{21}  &   -1  
\end{bmatrix}
\begin{bmatrix}
	x_1  \\
	x_2  
\end{bmatrix}=
-\begin{bmatrix}
	b_1 \\
	b_2 
\end{bmatrix}$, $b_1,b_2$ are uniform random numbers in $[0,1]$. Generate a random number $P$, which is uniform in $[0,1]$. With probability $1-P$, $a_{12}=0$, and with probability $P$, $a_{12}$ is a uniform random number in $[-1,1]$. The same applies to $a_{21}$. Solve $x_2$ and add noise: $x_2'=x_2(1+\epsilon)$, where $\epsilon$ is uniform in $[-N,N]$. Similarly, set $a_{21}=0$ and solve $\bar{x}_2$, then add noise: $\bar{x}_2'=\bar{x}(1+\epsilon')$, where $\epsilon'$ is independent and uniform in $[-N,N]$. To test the path blocking property, we check whether $1-T\le x_2'/\bar{x}_2'\le 1/(1-T)$. If this inequality holds, we infer that $V_1$ to $V_2$ can be blocked, and there is no edge $V_1\to V_2$. For each noise level $N\in[0,0.5]$ and threshold $T\in [0,1]$, the simulation is repeated for $10^6$ times.

\subsection{Path blocking method with a phenotype}
\label{apppbp}

For each $p\in[0,0.5]$ and $q\in[0,0.5]$, repeat the following simulation for $10^5$ times, and average the performance measures. In each simulation, generate a random number $R$ which is uniform in $[0,1]$. Generate a GRN, where each edge has probability $R$ to exist. Here an edge from the phenotype $V_0$ to a gene $V_i$ is not allowed. Use the GRN to generate path blocking relations from each $V_i$ to $V_0$, and apply Proposition \ref{pbpt} to obtain the inferred GRN $\mathcal{G}$. Perturb the path blocking relations: with probability $q$, switch blocking to not blocking; with probability $p$, switch not blocking to blocking. Apply Proposition \ref{pbpt} to obtain the inferred GRN $\mathcal{G}'$, and compare with $\mathcal{G}$ to calculate the four performance measures.

\subsection{Ancestor-descendant relation method}
\label{appadr}

For each $p\in[0,0.5]$ and $q\in[0,0.5]$, repeat the following simulation for $10^4$ times, and average the performance measures. In each simulation, generate a random integer $R$ which is uniform in $[10,30]$. Generate a GRN with $10$ genes and $R$ edges while guaranteeing it is a DAG. Use the DAG to generate ancestor-descendant relations, and apply Proposition \ref{adrp} to obtain the inferred GRN $\mathcal{G}$. Perturb the ancestor-descendant relations: with probability $p$, switch ``is a descendant'' to ``is not a descendant''; with probability $q$, switch ``is not a descendant'' to ``is a descendant''. Apply Proposition \ref{adrp} to obtain the inferred GRN $\mathcal{G}'$, and compare with $\mathcal{G}$ to calculate the four performance measures.

\subsection{Conditional independence + ancestor-descendant method}
\label{appci}

For each $p\in[0,0.5]$ and $q\in[0,0.5]$, repeat the following simulation for $10^5$ times, and average the performance measures. In each simulation, generate a random connected DAG with 3 genes. Since we require the GRN is a connected DAG, there are only four topologically different types: $\xymatrix{
	V_1\ar[rd]\ar[r]&V_2\ar[d]\\
	&V_3
}$, 
$\xymatrix{
	V_1\ar[r]&V_2\ar[d]\\
	&V_3
}$, 
$\xymatrix{
	V_1\ar[rd]&V_2\ar[d]\\
	&V_3
}$, 
$\xymatrix{
	V_1\ar[rd]\ar[r]&V_2\\
	&V_3
}$. For the generated DAG, determine the conditional independence relations and ancestor-descendant relations. Then we perturb the relations: with probability $q$, switch ``independent'' or ``is not an ancestor'' to ``dependent'' or ``is an ancestor''; with probability $p$, switch ``dependent'' or ``is an ancestor'' to ``independent'' or ``is not an ancestor''. With the perturbed conditional independence relations, we determine whether there is an edge between two vertices, and then we use the perturbed ancestor-descendant relations to determine the edge direction. If we determine there is an edge between $V_i$ and $V_j$, but the ancestor-descendant relations are contradicted (both $V_i$ and $V_j$ are an ancestor of the other, or neither $V_i$ nor $V_j$ is an ancestor of the other), we randomly assign a direction with equal probability. With this method, we obtain the inferred DAG, and compare with the true DAG to calculate the four performance measures.

\subsection{Partially observed ODE method}
\label{apppode}
We start from the linear ODE system 
\begin{equation*}
	\begin{split}
		\mathrm{d}x_2/\mathrm{d}t &=-3x_2+1,\\
		\mathrm{d}x_1/\mathrm{d}t &=x_2-2x_1+1,\\
		\mathrm{d}x_0/\mathrm{d}t &=2x_2+x_1-x_0+1,
	\end{split}
\end{equation*}
whose stationary state is $[x_2,x_1,x_0]=[1/3,2/3,7/3]$. If the system starts from $[4/3,2/3,7/3]$, then 
\[x_0(t)=\frac{3}{2}e^{-t}-e^{-2t}-\frac{1}{2}e^{-3t}+\frac{7}{3}.\]
In the equation $x_0(t)=c_0e^{a_{00}t}+c_1e^{a_{11}t}+c_2e^{a_{22}t}+7/3$, set $t=1,2,3,4,5,6$, and solve $a_{00},a_{11},a_{22}$ with the fsolve function in Matlab. When the true $x_0(t)$ is replaced by $(1+\epsilon)x_0(t)$, where $\epsilon$ is a random noise that is uniform on $[-0.01,0.01]$, the solver fails to find a solution. We repeat this procedure after applying the Laplace transform on $x_0(t)$, and the solver still fails. The problem of solving $a_{00},a_{11},a_{22}$ is essentially sensitive to perturbations. In reality, even with a small noise, this method fails.

\subsection{Path blocking + ancestor-descendant method}
\label{apppba}

For each $p\in[0,0.5]$ and $q\in[0,0.5]$, repeat the following simulation for $10^5$ times, and average the performance measures. In each simulation, generate a random number $R$ which is uniform in $[0,1]$. Generate a GRN, where each edge has probability $R$ to exist. Here an edge from the phenotype $V_0$ to a gene $V_i$ is not allowed. Make sure the GRN is a DAG. Use the DAG to generate path blocking relations from each $V_i$ to $V_0$, and ancestor-descendant relations between $V_i$ and $V_j$. First apply Proposition \ref{pbpt} to obtain the inferred DAG $\mathcal{G}$. Edges that cannot be determined will be checked by Proposition \ref{adrp}. Perturb the path blocking relations and the ancestor-descendant relations: with probability $q$, switch blocking to not blocking, or switch ``is not an ancestor'' to ``is an ancestor''; with probability $p$, switch not blocking to blocking, or switch ``is an ancestor'' to ``is not an ancestor''. Apply Proposition \ref{pbpt} to obtain the inferred GRN $\mathcal{G}'$, Edges that cannot be determined will be checked by Proposition \ref{adrp}. If Proposition \ref{pbpt} and Proposition \ref{adrp} have opposite results, rely on Proposition \ref{pbpt}. Compare $\mathcal{G}$ and $\mathcal{G}'$ to calculate the four performance measures.

\bibliographystyle{vancouver}
\bibliography{GRN}

\end{document}